\documentclass[10pt]{article}

\usepackage[paper=a4paper,left=23mm,right=23mm,top=28mm, bottom=25mm]{geometry}

\usepackage{hyperref}

\usepackage{authblk}
\usepackage{tikz}
\usepackage{xfrac}
\usepackage{tabularx, multirow}
\usepackage{graphicx, tikz}
\usetikzlibrary{arrows.meta}
\usetikzlibrary{arrows,calc}
\usetikzlibrary{shapes}
\usepackage{latexsym,amssymb,amsfonts,graphicx,makeidx,paralist}
\usepackage{theorem,rotating,ifthen,amsmath,subfigure,epsfig,nicefrac}
\usepackage{booktabs,framed}

\usepackage{hyperref}

\usepackage{xcolor}
\colorlet{shadecolor}{gray!12}

\newcommand{\g} {\mbox{digraph}}
\newcommand{\SPD} {\text{SPD}}
\newcommand{\MSP} {\text{MSP}}

\newcommand{\BC} {\text{BC}}

\newtheorem{remark}{Remark}

\newtheorem{proposition}{Proposition}

\newenvironment{desctight}
  {\begin{list}{}{\setlength\labelwidth{0pt}
        \setlength{\itemsep}{0.5pt}
        \setlength{\parsep}{0pt}
        \setlength\itemindent{-\leftmargin}
        }}
    {\end{list}}

\newtheorem{theorem}{Theorem}[section]
\newtheorem{example}[theorem]{Example}
\newtheorem{corollary}[theorem]{Corollary}
\newtheorem{lemma}[theorem]{Lemma}
\newtheorem{observation}[theorem]{Observation}

\newtheorem{definition}[theorem]{Definition}

\newenvironment{proof}{\noindent{\bf Proof~}}{\null\hfill $\Box$\par\medskip}

\newcommand{\bigo}{\text{$\mathcal O$}}

\newcommand{\ideg}{\text{indegree}}
\newcommand{\odeg}{\text{outdegree}}

\newcommand{\DC} {\text{DC}}
\newcommand{\OC} {\text{OC}}

\newcommand{\TT} {\text{TT}}
\newcommand{\SPO} {\text{SPO}}

\begin{document}

\title{Solutions for Subset Sum Problems with Special Digraph Constraints\thanks{A short version of 
this paper will appear in Proceedings of the {\em International Conference on Operations Research} (OR 2019) \cite{GKR19e}.}}

\author{Frank Gurski}
\author{Dominique Komander}
\author{Carolin Rehs}

\affil{\small University of  D\"usseldorf,
Institute of Computer Science, Algorithmics for Hard Problems Group,\newline 
40225 D\"usseldorf, Germany}

\maketitle


\begin{abstract}
The subset sum problem is one of the simplest and most fundamental
NP-hard problems in combinatorial optimization.
We consider two extensions of this problem:
The subset sum problem with digraph
constraint (SSG) and subset sum problem with weak digraph constraint (SSGW).
In both problems there is given a digraph with sizes assigned to the vertices.
Within SSG we want to find a subset of vertices
whose total size does not exceed a given capacity and which contains a vertex
if at least one of its  predecessors is part of the solution.
Within SSGW we want to find a subset of vertices
whose total size does not exceed a given capacity and which contains a vertex
if all its predecessors are part of the solution. SSG and SSGW have been
introduced recently by Gourv\`{e}s et al.\ who studied  their complexity for 
directed acyclic graphs and oriented trees.
We show that both problems are
NP-hard even on oriented co-graphs and minimal series-parallel digraphs. Further, 
we provide pseudo-polynomial solutions for SSG and SSGW with
digraph constraints given by directed co-graphs and series-parallel digraphs.

\end{abstract}


%


\section{Introduction}

The subset sum problem is one of the most fundamental
NP-hard problems in combinatorial optimization.
Within the {\em subset sum problem (SSP)}
there is given a set $A=\{a_1,\ldots,a_n\}$ of $n$ items.
Every item $a_j$ has a  size $s_j$ and there is a capacity $c$.
All values are assumed to be positive integers and $s_j\leq c$ for every
$j\in\{1,\ldots,n\}$.
The task is to choose a subset $A'$ of $A$, such that the sum of the sizes 
of the items in $A'$ is maximized and is at most $c$.

In this paper we consider the following two sum problems which additionally 
have given a digraph on the item set. Both problems have been
introduced recently by Gourv\`{e}s et al.\ \cite{GMT18}.
Within the {\em subset sum problem with digraph
constraint} (SSG) we want to find a subset of vertices
whose total size does not exceed a given capacity and which contains a vertex
if at least one of its  predecessors is part of the solution.
Within the {\em subset sum problem with weak digraph constraint} (SSGW) the goal is to 
find a subset of vertices whose total size does not exceed a given capacity and 
which contains a vertex if all its predecessors are part of the solution. 
Since SSG and SSGW generalize SSP, they are NP-hard.
Both problems are integer-valued problems, which motivates
to observe whether they are weakly NP-hard, i.e. 
the existence of pseudo-polynomial algorithms.

For related works we refer to \cite[Section 3]{GMT18}.
In \cite{GMT18} it has been shown that on directed acyclic graphs (DAGs) SSG is
strongly NP-hard and SSGW is even APX-hard. Further, they showed that the
restriction to oriented trees allows to give a  pseudo-polynomial algorithm
using  dynamic programming along the tree.

In this paper we consider SSG and SSGW on further special digraph classes.
First, we consider co-graphs (short for complement reducible graphs), which  
can be generated from the single vertex graph by applying
disjoint union, order composition and series composition \cite{BGR97}.
They can also  be characterized by excluding eight forbidden induced 
subdigraphs. 
Directed co-graphs are exactly the digraphs of
directed NLC-width\footnote{The abbreviation NLC results from the 
node label controlled embedding mechanism originally defined for graph
grammars \cite{ER97}.} 1 and a proper subset of the digraphs
of directed clique-width at most 2 \cite{GWY16}.
Directed co-graphs are interesting from an algorithmic point of view 
since several hard graph problems can be solved in 
polynomial time  by dynamic programming along the tree structure of
the input graph, see \cite{BM14,Gur17a,GR18c,GKR19f,GKR19d}.
Moreover, directed co-graphs are very useful for the reconstruction
of the evolutionary history of genes or species using genomic
sequence data \cite{HSW17,NEMWH18}.

Further, we look at SSG and SSGW
on series-parallel digraphs, which are exactly the digraphs whose transitive closure equals
the transitive closure of some minimal series-parallel digraph.
Minimal series-parallel digraphs can be generated from the single vertex graph by applying
the parallel composition and series composition \cite{VTL82}.
Series-parallel digraphs are also interesting from an algorithmic point of view 
since several hard graph problems can be solved in 
polynomial time  by dynamic programming along the tree structure of
the input graph, see \cite{MS77,Ste85,Ren86}.

We show pseudo-polynomial solutions for SSG and SSGW on directed
co-graphs and minimal  series-parallel digraphs and deduce
a pseudo-polynomial time solution for SSG on series-parallel digraphs.
Our results are based on dynamic programming along the
tree-structure of the considered digraphs. 
The considered digraph classes are incomparable w.r.t.\ inclusion
to oriented trees considered in  \cite{GMT18}, see Fig.\ \ref{grcl}. Moreover, 
the digraphs of our interest allow to define dense graphs, i.e.\ graphs where the number 
of directed edges is quadratic in the number of vertices.
In Table \ref{table-A} we summarize the known results from \cite{GMT18} and the results
of this work about subset sum problems with special digraph constraints.

\begin{table}[h]
\caption{Known running times for SSG and SSGW with digraph constraints restricted
to special graph classes. 
Let $n$ be the number of vertices and $m$ the number of directed edges of the input digraph and $c$
be the capacity.}\label{table-A}
\begin{center}
\begin{tabular}{l|p{2.9cm}p{2.3cm}|p{2.29cm}p{2.3cm}|}
                       & \multicolumn{2}{c|}{SSG}                     &   \multicolumn{2}{c|}{SSGW}        \\
\hline
transitive tournaments &  $\bigo(n^2)$                  &  Remark \ref{tt}           &   $\bigo(n\cdot c^4 + m)$  &Theorem   \ref{ssgw}     \\
bioriented cliques     &  $\bigo(n)$                    & Remark \ref{cl}            &    $\bigo(n\cdot c^4 + m)$  &Theorem   \ref{ssgw}     \\

\hline
DAGs                   &  strongly NP-hard              &\cite{GMT18}                &  APX-hard  &\cite{GMT18}              \\
oriented trees         &   $\bigo(n\cdot c^3)$          & \cite{GMT18}               & $\bigo(n\cdot c^2)$ & \cite{GMT18}             \\
directed co-graphs     &  $\bigo(n\cdot c^2 + m)$       & Theorem \ref{th-co-ssg}    &  $\bigo(n\cdot c^4 + m)$ & Theorem   \ref{ssgw}                \\
minimal series-parallel&  $\bigo(n\cdot c^2+m)$         & Theorem \ref{th-co-ssg-msp}&   $\bigo(n\cdot c^4+m)$ &  Theorem \ref{th-co2-ssg-sp}                 \\
series-parallel        &  $\bigo(n\cdot c^2+n^{2.37})$  &Theorem \ref{th-co-ssg-sp}                      &    open &         \\
\hline
\end{tabular}
\end{center}
\end{table}


\section{Preliminaries}

\subsection{Digraphs}

A {\em directed graph} or {\em digraph} is a pair  $G=(V,E)$, where $V$ is
a finite set of {\em vertices} and  $E\subseteq \{(u,v) \mid u,v \in
V,~u \not= v\}$ is a finite set of ordered pairs of distinct
vertices called {\em arcs} or {\em directed edges}.
For a vertex $v\in V$, the sets $N_G^+(v)=\{u\in V~|~ (v,u)\in E\}$ and 
$N_G^-(v)=\{u\in V~|~ (u,v)\in E\}$ are called the {\em set of all successors} 
and the {\em set of all  predecessors} of $v$ in $G$. 
The  {\em outdegree} of $v$, $\odeg_G(v)$ for short, is the number
of successors of $v$ and the  {\em indegree} of $v$, $\ideg_G(v)$ for short, 
is the number of predecessors of $v$.
We may omit indices if the graph under consideration is clear from the context.
A vertex  $v\in V$ is  {\em out-dominating (in-dominated)} if
it is adjacent to every other vertex in $V$ and is a source (a sink, respectively).

A digraph $G'=(V',E')$ is a {\em subdigraph} of digraph $G=(V,E)$ if $V'\subseteq V$
and $E'\subseteq E$.  If every arc of $E$ with start- and end-vertex in $V'$  is in
$E'$, we say that $G'$ is an {\em induced subdigraph} of digraph $G$ and we
write $G'=G[V']$.

For  $n \ge 2$ we denote by 
$$\overrightarrow{P_n}=(\{v_1,\ldots,v_n\},\{ (v_1,v_2),\ldots, (v_{n-1},v_n)\})$$
a {\em directed path} on $n$ vertices. Vertex $v_1$ is the {\em start vertex}
and $v_n$ is the {\em end vertex} of $\overrightarrow{P_n}$.
For  $n \ge 2$ we denote by
$$\overrightarrow{C_n}=(\{v_1,\ldots,v_n\},\{(v_1,v_2),\ldots, (v_{n-1},v_n),(v_n,v_1)\})$$
a {\em directed cycle} on $n$ vertices.

A {\em directed acyclic graph (DAG for short)} is a digraph without any $\overrightarrow{C_n}$,
for $n\geq 2$, as subdigraph.
A vertex $v$ is {\em reachable} from a vertex $u$ in $G$, if $G$ contains
a $\overrightarrow{P_n}$ as a subdigraph having start vertex $u$ and
end vertex $v$.

A {\em weakly connected component} of $G$ is a maximal subdigraph, 
such that the corresponding underlying graph is connected.
A {\em strongly connected component} of $G$ is a maximal subdigraph, in which
every vertex is reachable from every other vertex.

An  {\em out-rooted-tree} ({\em in-rooted-tree}) 
 is an orientation of a tree with a distinguished root such that
all arcs are directed away from (to) the root.

\subsection{Problems}

Let $A=\{a_1,\ldots,a_n\}$ be a set of $n$ items, such that every item $a_j$ has a  size $s_j$.
For a subset $A'$ of  $A$  we define
$$s(A'):=\sum_{a_j\in A'}s_j$$  and  the
{\em capacity constraint} by
\begin{equation}
s(A')\leq c. \label{cap}
\end{equation}

\begin{desctight}
\item[Name] Subset sum  (SSP)

\item[Instance]
A set $A=\{a_1,\ldots,a_n\}$ of $n$ items.
Every item $a_j$ has a  size $s_j$ and there
is a capacity $c$.

\item[Task]  Find a  subset $A'$ of $A$
that maximizes $s(A')$ subject to (\ref{cap}).
\end{desctight}

The parameters $n$, $s_j$, and $c$ are assumed to be 
positive integers. See \cite[Chapter 4]{KPP10} for a survey on the subset sum problem.
In order to consider generalizations of the subset sum problem we will
consider  constraints for a digraph $G=(A,E)$ with objects assigned
to the vertices.

The {\em digraph constraint} ensures that $A'\subseteq A$  contains a vertex $y$, if
it contains at least one predecessor of $y$, i.e.\
\begin{equation}
\forall y\in A \left(N^{-}(y)\cap A' \neq   \emptyset  \right) \Rightarrow  y \in A'. \label{dc}
\end{equation}

The {\em weak digraph constraint} ensures that
$A'$ contains a vertex $y$, if
it contains every predecessor of $y$, i.e.\
\begin{equation}
\forall y\in A \left(N^{-}(y)\subseteq A' \wedge N^{-}(y) \neq \emptyset \right) \Rightarrow y \in A'. \label{wdc}
\end{equation}

This allows us to state the following optimization problems given in \cite{GMT18}.

\begin{desctight}
\item[Name] Subset sum  with  digraph constraint (SSG)

\item[Instance]
A set $A=\{a_1,\ldots,a_n\}$ of $n$ items and a digraph $G=(A,E)$.
Every item $a_j$ has a  size $s_j$ and there
is a capacity $c$.

\item[Task]  Find a  subset $A'$ of $A$
that maximizes $s(A')$ subject to (\ref{cap}) and (\ref{dc}).
\end{desctight}

\begin{desctight}
\item[Name] Subset sum  with weak digraph constraint (SSGW)

\item[Instance] A set $A=\{a_1,\ldots,a_n\}$ of $n$ items and a digraph $G=(A,E)$.
Every item $a_j$ has a  size $s_j$ and there is a capacity $c$.

\item[Task] Find a  subset $A'$ of $A$
that maximizes $s(A')$ subject to  (\ref{cap}) and (\ref{wdc}).
\end{desctight}

In our problems the parameters $n$, $s_j$, and $c$ are assumed to be 
positive integers.\footnote{The results in \cite{GMT18} also consider null sizes,
which are excluded in our work. All our solutions
can be extended to  pseudopolynomial solutions which solve  SSG and SSGW using null sizes, 
see Section \ref{sec-con}.}
Further, in the defined problems a subset $A'$ of $A$ is called {\em feasible}, if
it satisfies the prescribed constraints of the problem.
By $OPT(I)$ we denote the value of an
optimal solution for input $I$.

\medskip
\begin{observation}
Every feasible solution for SSG is also a feasible solution for
SSGW, but not vice versa. 
\end{observation}

\medskip
\begin{observation}
$A'=\emptyset$ and $A'=A$ for $s(A)\leq c$ are feasible solutions for every instance of SSG and for every instance of 
SSGW. 
\end{observation}

\medskip
In order to give equivalent characterizations for SSG and SSGW
we use binary integer programs.

\begin{remark}
To formulate SSG and SSGW as a binary integer program,
we introduce a binary variable $x_j\in \{0,1\}$ for each 
item $a_j\in A$, $1\leq j \leq n$. The idea is to have 
$x_j=1$ if and only if item $a_j\in A'$. 
\begin{enumerate}
\item  SSG corresponds
to maximizing  $\sum_{j=1}^n s_jx_j$   subject to $\sum_{j=1}^n s_j x_j \leq c$,
$x_i\leq x_j$ for every $j\in\{1,\ldots n\}$ and for every $a_i\in N^{-}(a_j)$, and $x_j\in\{0,1\}$
for every $j\in\{1,\ldots n\}$.

\item  SSGW corresponds
to maximizing  $\sum_{j=1}^n s_jx_j$   subject to   $\sum_{j=1}^n s_j x_j \leq c$, 
$\sum_{\{i \mid a_i\in N^{-}(a_j)\}} x_i \leq x_j + \ideg(a_j) - 1$
for every $j\in\{1,\ldots n\}$, and $x_j\in\{0,1\}$
for every $j\in\{1,\ldots n\}$.
\end{enumerate}
\end{remark}

The complexity for SSG and SSGW restricted to DAGs and oriented trees was considered
in \cite{GMT18}.

\begin{theorem}[\cite{GMT18}]
On DAGs SSG is strongly NP-hard and  SSGW is APX-hard.
\end{theorem}

\subsection{Basic results}

Let $G=(V,E)$ be a digraph  and $x\in V$.
By $R_x$ we denote the vertices
of $V$ which are reachable from $x$
and by  $S_x$ we denote the vertices
of $V$ which are in the same strongly connected component as $x$.
Thus, it holds that $\{x\}\subseteq S_x \subseteq R_x \subseteq V$.

\begin{lemma}\label{le-reach-fr}
Let $A'$ be a feasible solution for SSG
on a digraph $G=(A,E)$ and $x\in A$. Then, it holds that
$x\in A'$ if and only if $R_x\subseteq A'$.
\end{lemma}

\begin{lemma}\label{le-reach-sc}
Let $A'$ be a feasible solution for SSG
on a digraph
$G=(A,E)$ and $x\in A$. Then, it holds that
$x\in A'$ if and only if $S_x\subseteq A'$.
\end{lemma}

\begin{lemma}\label{cl1c}
SSG is solvable in  $\bigo(2^t\cdot (n+m))$ time on digraphs with $n$ 
vertices, $m$ arcs, and $t$ strongly connected components.
\end{lemma}

\begin{proof}
By Lemma \ref{le-reach-sc} for every feasible solution $A'$ and
every  strongly connected component $S$, it either holds that $S\subseteq A'$
or $S\cap A'=\emptyset$. Since all strongly connected components
are vertex disjoint, we can solve SSG by verifying $2^t$ possible
feasible solutions. Verifying the capacity constraint can be done
in $\bigo(n)$ time and verifying the digraph constraint can be done in
$\bigo(n+m)$ time. 
\end{proof}

In the condensation $con(G)$ of a digraph  $G=(V,E)$
every strongly connected component $C$ of $G$  
is represented by a vertex $v_C$ and there is
an arc between two
vertices $v_C$ and $v_{C'}$ if there exist $u\in C$ and $v\in C'$, such
that $(u,v)\in E$. For every digraph $G$ it holds that $con(G)$
is a directed acyclic graph.

In order to solve SSG it is useful to consider the condensation of the 
input digraph $G=(A,E)$. By defining the size
of a vertex  $v_C$ of $con(G)$ by the 
sum of the sizes of the vertices in $C$,
the following result has been shown in \cite[Lemma 2]{GMT18}.

\begin{lemma}[\cite{GMT18}]\label{lem1}
For a given instance of SSG on digraph $G$, there is a bijection between the 
feasible solutions (and thus the set of optimal
solutions) of SSG for $G$ and the 
feasible solutions (and thus the set of optimal
solutions) for $con(G)$.
\end{lemma}

Thus, in order to solve SSG we can restrict ourselves to
DAGs by computing the condensation of the input
graph in a first step. 
%
The next example shows that Lemma \ref{lem1} does not hold
for SSGW.

\begin{figure}[hbtp]
\centering
\parbox[b]{60mm}{
\centerline{\epsfxsize=30mm \epsfbox{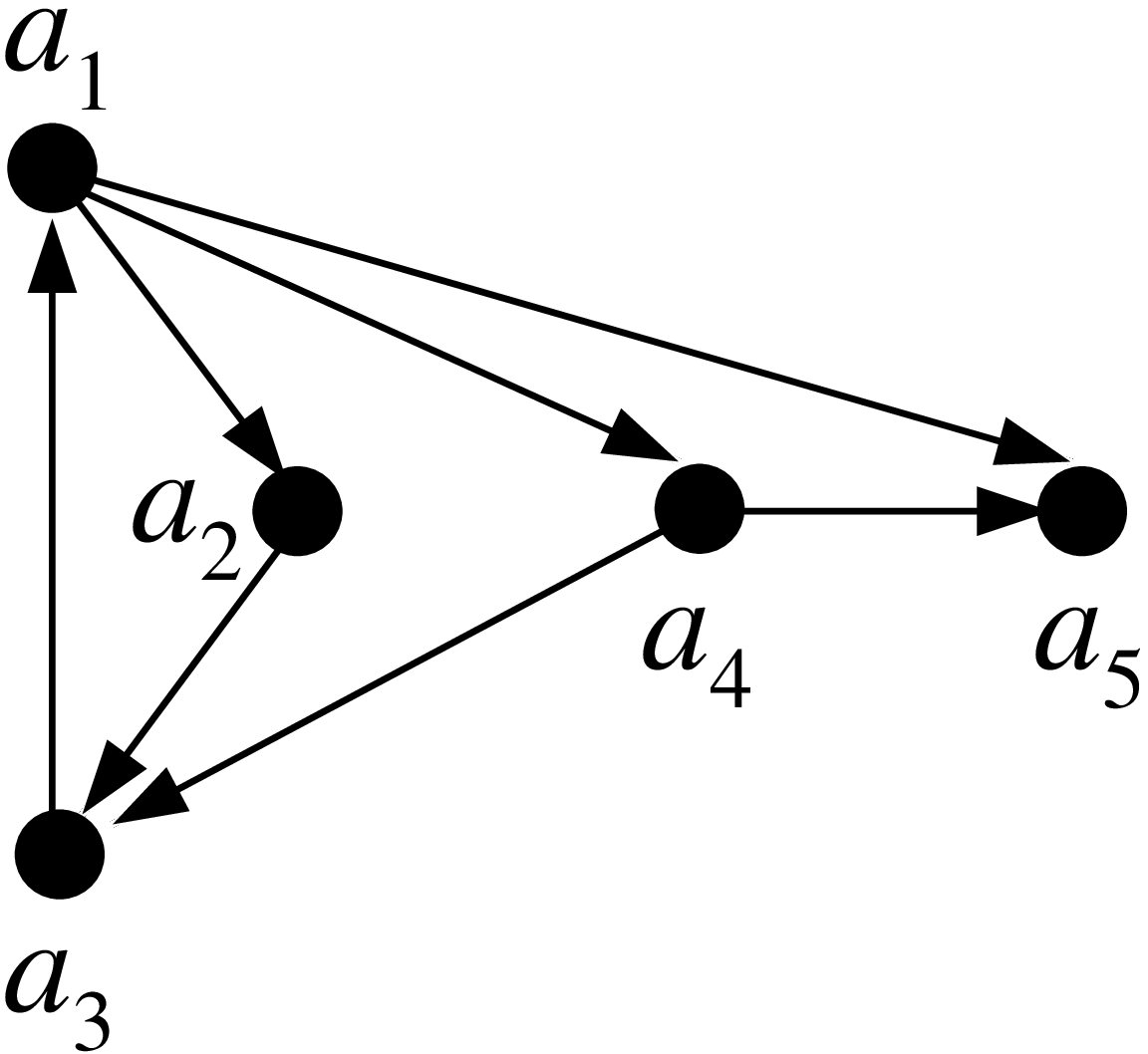}}
\caption{Digraph  in Example \ref{ex-ssgw-sc}.}
\label{F00}}
\hspace{1cm}
\parbox[b]{60mm}{
\centerline{\epsfxsize=30mm \epsfbox{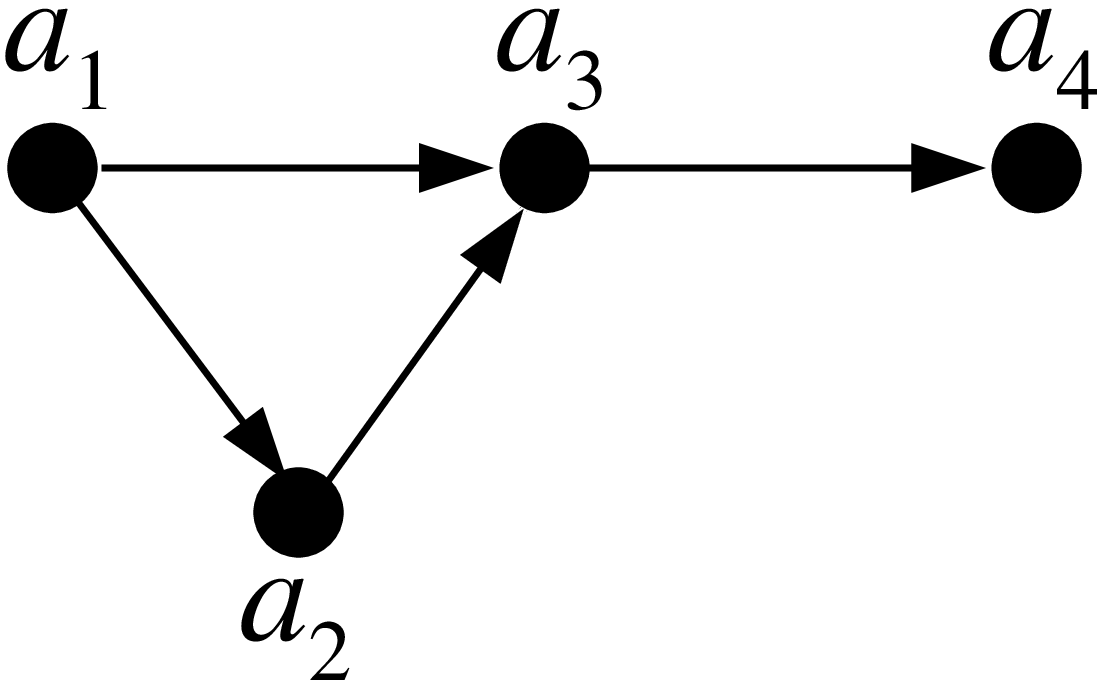}}
\caption{Digraph  in Example \ref{ex-ssgw-tr}.}
\label{F01}}
\end{figure}

\begin{example}\label{ex-ssgw-sc} We consider the digraph $G$ in
Fig.~\ref{F00}. For SSGW with $c=2$ and
all sizes $s_j=1$ we have among others $\{a_4\}$ as a feasible solution.
Since $con(G)$ is a path of length one, formally
$$con(G)=(\{v_{\{a_1,a_2,a_3,a_4\}},v_{\{a_5\}}\},\{(v_{\{a_1,a_2,a_3,a_4\}},v_{\{a_5\}})\}),$$
the only feasible solution is $\{a_5\}$, which implies that
$\{a_4\}$ is not a feasible solution for SSGW using
$con(G)$.
\end{example}

The transitive closure $td(G)$ of a digraph $G$ has the same
vertex set as $G$ and for two distinct vertices $u,v$ there is
an arc $(u,v)$ in $td(G)$ if and only if there is
a directed path from $u$ to $v$ in $G$. The transitive reduction $tr(G)$ of
a digraph $G$ has the same
vertex set as $G$ and as few arcs of $G$ as possible, such that
$G$ and $tr(G)$ have the same transitive closure.
The transitive closure is unique for every digraph.
The transitive reduction is unique for directed acyclic graphs.
However, for arbitrary digraphs the transitive reduction is not unique.
The time complexity of the best known algorithm
for finding the transitive reduction of a graph is the same as the time
to compute the transitive closure of a graph or to perform Boolean matrix
multiplication \cite{AGU72}.
The best known algorithm  to perform Boolean matrix multiplication
has running  time $\bigo(n^{2.3729})$ by \cite{Leg14}.

\begin{lemma} \label{lem2}
For a given instance of SSG on a directed acyclic
graph $G$, the set of feasible solutions and thus the set of
optimal solutions of SSG
for $G$ and for $tr(G)$ are equal.
\end{lemma}

\begin{proof}
Since a transitive reduction is a subdigraph of the given graph,
every feasible solution $A'$ for $G$ is also a feasible solution
for the transitive reduction $tr(G)$.
To show the reverse direction,
let $A'$ be a feasible solution for $tr(G)$. By the definition
of $tr(G)$ we know that for every vertex $v$, every predecessor $u$ of $v$
in $G$ is also a predecessor of $v$ in $tr(G)$ or there is a path from
$u$ to $v$ in $tr(G)$. By Lemma \ref{le-reach-fr}
we know that $A'$ is also a feasible solution for $G$. 
\end{proof}

Thus, in order to solve SSG we can restrict ourselves to transitive reductions.
The next example shows that Lemma \ref{lem2} does not hold
for SSGW.

\begin{example}\label{ex-ssgw-tr}
We consider the digraph $G$ in Fig.~\ref{F01}. For SSGW with $c=2$ and
all sizes $s_j=1$ we have among others $\{a_2\}$ as a feasible solution.
Since $tr(G)$ is a path, formally
$$tr(G)=(\{a_1,a_2,a_3,a_4\},\{(a_1,a_2),(a_2,a_3),(a_3,a_4)\}),$$
$a_2$ implies by (\ref{wdc}) that $a_3$ and $a_4$
must be part of the  solution, which implies that
$\{a_2\}$ is not a feasible solution for SSGW using
$tr(G)$.
\end{example}

In the correctness proofs of our algorithms in Sections \ref{sol-co}
and \ref{sol-spd}
we will use the following lemmata.

\begin{lemma}\label{induced-sd}
Let $G=(V_G,E_G)$ be a digraph and let $H=(V_H,E_H)$ be an induced
subdigraph of $G$. If $A'$ is a feasible solution for SSG on $G$,
then $A'\cap V_H$ is  a feasible solution for SSG on $H$. 
\end{lemma}

\begin{proof}
 If $A'$ is a feasible solution for SSG on $G$, then it holds that
$$
\forall y\in V_G \left(N_G^{-}(y)\cap A' \neq   \emptyset  \right) \Rightarrow  y \in A'.
$$
By restricting to  $y$ having no predecessors  from $V_G\setminus V_H$, we obtain
%
%
$$
\forall y\in V_G \left(N_G^{-}(y)\cap A'\cap V_H \neq   \emptyset   \right) \Rightarrow  y \in A'.
$$
By restricting $y$ to $V_H\subseteq V_G$ we obtain 
$$
\forall y\in V_H \left(N_H^{-}(y)\cap A'\cap V_H \neq   \emptyset   \right) \Rightarrow  y \in A'\cap V_H,
$$
i.e., $A'\cap V_H$ is  a feasible solution for SSG on $H$.

\end{proof}

The reverse direction of Lemma \ref{induced-sd} does not hold, since
vertices with predecessors in $A'\cap (V_G\setminus V_H)$ are not 
considered by the feasible solutions for  SSG on $H$. 
By considering  the induced subdigraph $H=(\{a_2,a_3,a_4\},\{(a_2,a_3),(a_3,a_4)\})$ of 
digraph $G$ in Example \ref{ex-ssgw-tr} we observe that 
Lemma \ref{induced-sd} does not hold for SSGW.
%
%
%
%
%
%
Next, we give two weaker forms of Lemma \ref{induced-sd}
which also hold for SSGW.  

\begin{lemma}\label{induced-sd3}
Let $G=(V_G,E_G)$ be a digraph and let $H=(V_H,E_H)$ be a weakly connected
component of $G$. If $A'$ is a feasible solution for SSGW on $G$,
then $A'\cap V_H$ is  a feasible solution for SSGW on $H$. 
\end{lemma}

\begin{proof}
 If $A'$ is a feasible solution for SSGW on $G$, then it holds that
$$
 \forall y\in V_G  \left(N_G^{-}(y)\subseteq A' \wedge N_G^{-}(y) \neq \emptyset \right) \Rightarrow y \in A'. 
$$
By restricting $y$ to $V_H\subseteq V_G$ we obtain 
$$
 \forall y\in V_H  \left(N_G^{-}(y)\subseteq A' \wedge N_G^{-}(y) \neq \emptyset \right) \Rightarrow y \in A'\cap V_H. 
$$
Since $H$ is a weakly connected
component of $G$ for all $y\in V_H$ it holds that $N_H^{-}(y)= N_G^{-}(y)$ such that
$$
 \forall y\in V_H  \left(N_H^{-}(y)\subseteq A'\cap V_H\wedge N_H^{-}(y) \neq \emptyset \right) \Rightarrow y \in A'\cap V_H,
$$
i.e., $A'\cap V_H$ is  a feasible solution for SSGW on $H$.
\end{proof}

\begin{lemma}\label{induced-sd4}
Let $G=(V_G,E_G)$ be a digraph and let $H=(V_H,E_H)$ be an induced 
subdigraph of $G$, such that no non-source of $H$ has a predecessor
in $V_G \setminus V_H$.
If $A'$ is a feasible solution for SSGW on $G$,
then $A'\cap V_H$ is  a feasible solution for SSGW on $H$. 
\end{lemma}

\begin{proof}
 If $A'$ is a feasible solution for SSGW on $G$, then it holds that
$$
 \forall y\in V_G  \left(N_G^{-}(y)\subseteq A' \wedge N_G^{-}(y) \neq \emptyset \right) \Rightarrow y \in A'. 
$$
By restricting $y$ to $V_H\subseteq V_G$ we obtain that 
$$
 \forall y\in V_H  \left(N_G^{-}(y)\subseteq A' \wedge N_G^{-}(y) \neq \emptyset \right) \Rightarrow y \in A'\cap V_H. 
$$
By restricting $y$ to be a non-source of $H$, we obtain
%
$$
 \forall y\in V_H  \left( N_G^{-}(y)\subseteq A' \wedge N_G^{-}(y) \neq \emptyset  \wedge  N_H^{-}(y)\neq \emptyset  \right) \Rightarrow y \in A'\cap V_H. 
$$
Since  no non-source of $H$ has a predecessor
in $V_G \setminus V_H$, we obtain
$$
 \forall y\in V_H  \left( N_G^{-}(y)\subseteq A'  \wedge  N_H^{-}(y)\neq \emptyset \right) \Rightarrow y \in A'\cap V_H. 
$$
Then it holds in $H$ that
$$
 \forall y\in V_H  \left( N_H^{-}(y)\subseteq A' \cap V_H \wedge  N_H^{-}(y)\neq \emptyset  \right) \Rightarrow y \in A'\cap V_H, 
$$
i.e., $A'\cap V_H$ is  a feasible solution for SSGW on $H$.
\end{proof}

\begin{lemma}\label{induced-sd5}
Let $G=(V_G,E_G)$ be a digraph such that there is a 2-partition $(V_1,V_2)$ of $V_G$ with
$\{(u,v) \mid u\in V_1, v\in V_2\}\subseteq E_G$.
If $A'$ is a feasible solution for SSGW on $G$ such that $V_1\subseteq A'$,
then $A'\cap V_2$ is  a feasible solution for SSGW on $G[V_2]$. 
\end{lemma}

\begin{proof}
 If $A'$ is a feasible solution for SSGW on $G$, then it holds that
$$
 \forall y\in V_G  \left(N_G^{-}(y)\subseteq A' \wedge N_G^{-}(y) \neq \emptyset \right) \Rightarrow y \in A'. 
$$
By restricting $y$ to $V_2\subseteq V_G$ we obtain 
$$
 \forall y\in V_2  \left(N_G^{-}(y)\subseteq A' \wedge N_G^{-}(y) \neq \emptyset \right) \Rightarrow y \in A'\cap V_2. 
$$
Since $V_1\subseteq A'$ it holds that
$$
 \forall y\in V_2  \left(N_G^{-}(y)\subseteq V_1 \cup A'\cap V_2 \wedge N_G^{-}(y) \neq \emptyset \right) \Rightarrow y \in A'\cap V_2. 
$$
Thus, it holds that
$$
 \forall y\in V_2  \left(N_{G[V_2]}^{-}(y)\subseteq  A'\cap V_2 \wedge (N_{G[V_1]}^{-}(y) \cup N_{G[V_1]}^{-}(y) )\neq \emptyset \right) \Rightarrow y \in A'\cap V_2. 
$$
Since  $V_1= N_{G[V_1]}^{-}(y)\neq \emptyset$ it holds that
$$
 \forall y\in V_2  \left(N_{G[V_2]}^{-}(y)\subseteq  A'\cap V_2  \right) \Rightarrow y \in A'\cap V_2. 
$$
By the properties of the logical implication it also holds that
$$
 \forall y\in V_2  \left(N_{G[V_2]}^{-}(y)\subseteq A'\cap V_2 \wedge N_{G[V_2]}^{-}(y) \neq \emptyset \right) \Rightarrow y \in A'\cap V_2, 
$$
i.e., $A'\cap V_2$ is  a feasible solution for SSGW on $G[V_2]$.
\end{proof}

Further, we will use the following result for solutions of SSP on digraphs
with sizes assigned to the vertices.

\medskip
\begin{observation}\label{induced-sd2}
Let $G=(V_G,E_G)$ be a digraph with sizes assigned to the vertices and 
let $H=(V_H,E_H)$ be an induced
subdigraph of $G$. If $A'\subseteq V_G$ satisfies (\ref{cap}),
then $A'\cap V_H$ satisfies (\ref{cap}). 
\end{observation}

\section{SSG and SSGW on directed co-graphs}\label{sol-co}

\subsection{Directed co-graphs}

Let $G_1=(V_1,E_1)$ and $G_2=(V_2,E_2)$ be two vertex-disjoint digraphs.
The following operations
have already been considered  by Bechet et al.\ in \cite{BGR97}.
\begin{itemize}
\item
The {\em disjoint union} of $G_1$ and $G_2$,
denoted by $G_1 \oplus  G_2$,
is the digraph with vertex set $V_1\cup V_2$ and
arc set $E_1\cup  E_2$.

\item
The {\em series composition} of $G_1$ and $G_2$,
denoted by $G_1\otimes  G_2$,
is defined by their disjoint union plus all possible arcs between
vertices of $G_1$ and $G_2$.

\item
The {\em order composition} of $G_1$ and $G_2$,
denoted by $G_1\oslash G_2$,
is defined by their disjoint union plus all possible arcs from
vertices of $G_1$ to vertices of $G_2$.
\end{itemize}

We recall the definition of directed co-graphs from \cite{CP06}.\footnote{In \cite{CP06} 
directed co-graphs are defined by disjoint union, series composition, and order composition
combining an arbitrary number of digraphs. We restrict ourselves to binary operations, which is
possible since these operations are associative.}

\begin{definition}[Directed co-graphs, \cite{CP06}]\label{dcog}
The class of {\em directed co-graphs} is recursively defined as follows.
\begin{enumerate}
\item Every digraph with a single vertex $(\{v\},\emptyset)$,
denoted by $v$, is a {\em directed co-graph}.

\item If  $G_1$ and $G_2$  are vertex-disjoint directed co-graphs, then
\begin{enumerate}
\item
the disjoint union
$G_1\oplus G_2$,

\item
the series composition
$G_1 \otimes G_2$, and
\item
the order composition
$G_1\oslash G_2$  are {\em directed co-graphs}.
\end{enumerate}
\end{enumerate}
The class of directed co-graphs is denoted by $\DC$.
\end{definition}

Every expression $X$ using  the four operations of Definition  \ref{dcog} 
is called a {\em di-co-expression} and
$\g(X)$ is the defined digraph.

\begin{example}\label{ex-dico}
The di-co-expression 
\begin{equation}
X=((v_1\oplus v_3) \oslash(v_2 \otimes v_4)) \label{eq-ori-c4}
\end{equation}
defines $\g(X)$ shown in Fig.~\ref{F02}.
\end{example}

\begin{figure}[hbtp]
\centering
\parbox[b]{60mm}{
\centerline{\epsfxsize=30mm \epsfbox{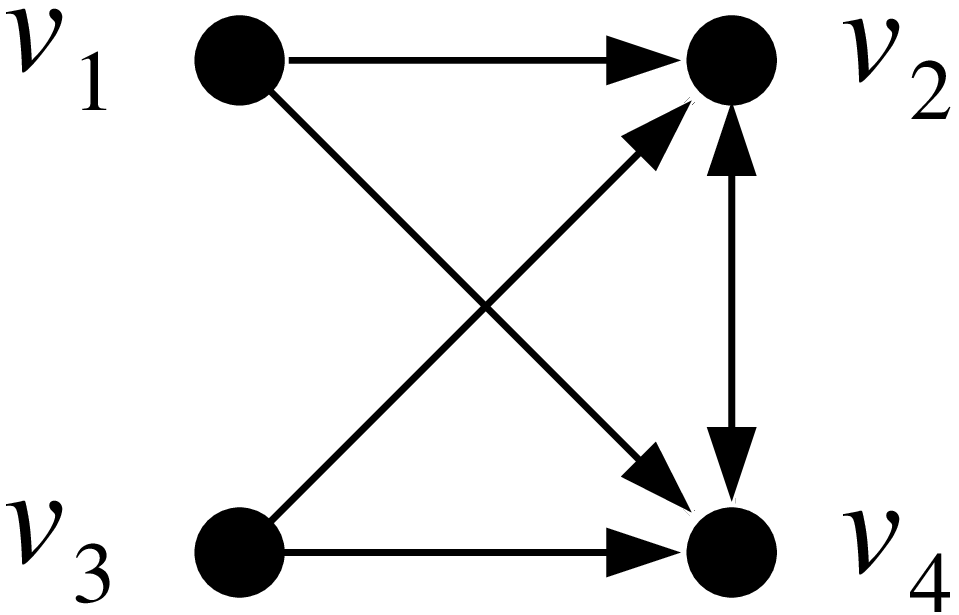}}
\caption{Digraph  in Example \ref{ex-dico}.}
\label{F02}}
\hspace{1cm}
\parbox[b]{60mm}{
\centerline{\epsfxsize=30mm \epsfbox{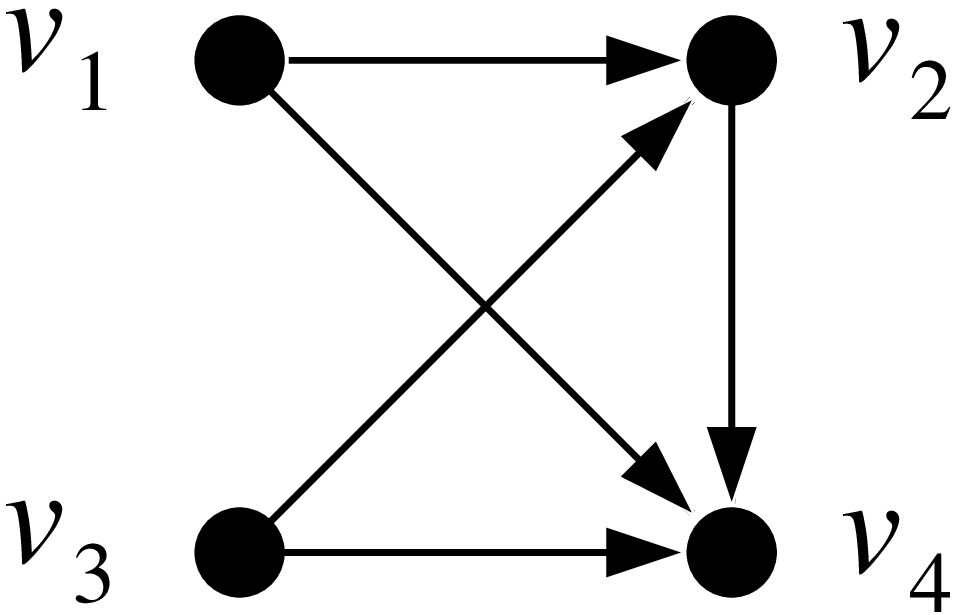}}
\caption{Digraph in Example \ref{ex-orico}.
}
\label{F03}}
\end{figure}

As undirected co-graphs can be characterized by forbidding the $P_4$,
directed co-graphs can be characterized likewise by
excluding eight forbidden induced subdigraphs \cite{CP06}.
For every directed co-graph we can define a tree structure
denoted as {\em di-co-tree}. It is a binary ordered rooted tree whose
vertices are labeled by the operations of the di-co-expression.

\begin{definition}[Di-co-tree]\label{dcot}
The {\em di-co-tree} for some directed co-graph $G$ is recursively defined as follows.
\begin{itemize}
\item The {\em di-co-tree} $T$ for di-co-expression $v$ consists of
a single vertex $r$ (the {\em root} of $T$) labeled by $v$.

\item The {\em di-co-tree} $T$ for   di-co-expression  $G_1\oplus G_2$ 
consists of a copy $T_1$ of the di-co-tree for $G_1$, 
a copy $T_2$ of the di-co-tree for $G_2$, 
an additional vertex $r$ (the {\em root}
of $T$) labeled by $\oplus$ and two additional arcs from vertex $r$ to the roots of $T_1$
and $T_2$. 
The root of $T_1$ is the {\em first child} of $r$ and 
the root of $T_2$ is the {\em second child} of $r$.

\item The {\em di-co-tree} $T$ for   di-co-expressions  $G_1\otimes G_2$ 
and  $G_1\oslash G_2$ are defined analogously to   $G_1\oplus G_2$.
\end{itemize}
\end{definition}

For every directed co-graph one can construct a di-co-tree in linear time,
see \cite{CP06}. Due to their recursive structure there are  problems 
that are hard in general but which can be solved efficiently
on directed 
co-graphs, see \cite{BM14,Gur17a,GKR19f,GKR19d,GR18c,GHKRRW20}.

\medskip
\begin{observation}\label{obs-red}
Let $G$ be a directed co-graph  and $T$ be a  di-co-tree
for $G$. For every vertex $u$ of $T$ which corresponds to a series operation,
the subtree rooted at $u$ defines a strongly connected subdigraph of $G$. 
Further, for every vertex $u$ of $T$ representing a series operation, such that
no predecessor of $u$ corresponds to a series operation, the
leaves of the subtree rooted at $u$ correspond to a strongly connected component of $G$.
\end{observation}

\medskip
By  omitting the series composition
within Definition \ref{dcog} we obtain the class of all
{\em oriented  co-graphs}. The class of oriented co-graphs is denoted by $\OC$.

\begin{example}\label{ex-orico}
The di-co-expression 
\begin{equation}
X=((v_1\oplus v_3) \oslash(v_2 \oslash v_4)) \label{eq-ori-c4x}
\end{equation}
defines $\g(X)$ shown in Fig.~\ref{F03}.
\end{example}

The class of oriented co-graphs has already been analyzed by Lawler
in \cite{Law76} and \cite[Section 5]{CLS81} using the notation
of {\em transitive series-parallel (TSP) digraphs}.
A digraph $G=(V,E)$ is called {\em transitive} if for
every pair $(u,v)\in E$ and $(v,w)\in E$ of arcs
with $u\neq w$ the arc $(u,w)$ also belongs to $E$.
For oriented co-graphs  the oriented chromatic number
and also the graph isomorphism problem can be solved in
linear time \cite{GKR19d}.

\medskip
\begin{observation}
Every oriented co-graph is a directed co-graph and every oriented co-graph  is a DAG.
\end{observation}

\medskip
Since SSP corresponds to SSG and also to SSGW on a digraph without arcs,
which is an oriented co-graph, we obtain the following result.

\begin{proposition}\label{np-oc}
SSG and SSGW are NP-hard on oriented co-graphs.
\end{proposition}

Next, we will show pseudo-polynomial solutions for SSG and SSGW restricted to
directed co-graphs. The main idea is a dynamic programming along 
the recursive structure of a given directed co-graph.

\subsection{Subset sum  with  digraph constraint (SSG)}

By Lemma \ref{lem1} in order to solve SSG we
can restrict ourselves to directed acyclic graphs.
This can be done by replacing every
strongly connected component $S$ by a new vertex $x_S$ whose size is the sum of
the sizes of the vertices in $S$. 
In order to identify the strongly connected components of directed co-graphs  using 
a di-co-tree we apply Observation \ref{obs-red}.
We perform a breadth first search on a di-co-tree $T$ starting at the root and for every vertex $u$
of $T$ which corresponds to a series operation we substitute the subtree
rooted at $u$ by a single vertex whose size is the sum of
the sizes of the vertices corresponding to the leaves of the subtree
rooted at $u$.
This does not reduce the size of the digraph or its di-co-tree in general, 
e.g. for oriented co-graphs
we have no non-trivial strongly connected component.

We consider an instance of SSG such that $G=(A,E)$ is a directed co-graph which
is given by some di-co-expression $X$.
For some subexpression $X'$ of $X$ let $F(X',s)=1$ if there is a solution
$A'$   in the graph defined by $X'$ satisfying (\ref{cap}) and (\ref{dc}) 
such that $s(A')=s$, otherwise let $F(X',s)=0$.
We use the notation $s(X')= \sum_{a_j\in X'}s_j$.

\begin{lemma}\label{le1} Let $0\leq s \leq c$.
\begin{enumerate}
\item $F(a_j,s)=1$ if and only if  $s=0$ or $s_j=s$.

In all other cases  $F(a_j,s)=0$.
\item $F(X_1\oplus X_2,s)=1$, if and only if 
there are  some $0\leq s'\leq s$ and $0\leq s''\leq s$
such that $s'+s''=s$ and $F(X_1,s')=1$ and $F(X_2,s'')=1$.

In all other cases  $F(X_1\oplus X_2,s)=0$.

\item $F(X_1\oslash X_2,s)=1$, if and only if 
\begin{itemize}
\item
$F(X_2,s)=1$ for $0\leq s \leq  s(X_2)$\footnote{The value $s=0$ is 
for choosing an empty solution in $\g(X_1\oslash X_2)$.} or

\item
there is an $s'>0$, such that $s=s'+s(X_2)$  and  $F(X_1,s')=1$.
\end{itemize}
In all other cases  $F(X_1\oslash X_2,s)=0$.

\item $F(X_1\otimes X_2,s)=1$, if and only if $s=0$ or $s=s(X_1)+s(X_2)$.

In all other cases  $F(X_1\otimes X_2,s)=0$.
\end{enumerate}
\end{lemma}

\begin{proof}
We show the correctness of the stated equivalences. Let $0\leq s \leq c$.
\begin{enumerate}
\item The only possible solutions in $\g(a_j)$  are $\emptyset$ and $\{a_j\}$ which have size 
$0$ and $s_j$, respectively.

\item If $F(X_1  \oplus X_2,s)=1$, then by Lemma \ref{induced-sd} there are  
$s'$ and $s''$ such that $s'+s''=s$ and solutions in $\g(X_1)$  
and in $\g(X_2)$ which guarantee $F(X_1,s')=1$ and $F(X_2,s'')=1$.

\medskip
Further, for every $s'$ and $s''$, such that $s'+s''=s$, $F(X_1,s')=1$, and $F(X_2,s'')=1$,
it holds  that $F(X_1  \oplus X_2,s)=1$ since the operation (disjoint union) 
does not create  new edges.

\item If $F(X_1  \oslash X_2,s)=1$, then we distinguish two cases.
If the solution of size $s$ in $\g(X_1\oslash X_2)$ contains no 
vertices of  $\g(X_1)$, then by Lemma \ref{induced-sd} there is a solution in $\g(X_2)$ which
guarantees $F(X_2,s)=1$.  

If the solution $A'$ of size $s$ in $\g(X_1\oslash X_2)$ contains at least one vertex
of  $\g(X_1)$, then by (\ref{dc}) solution $A'$ has to contain all vertices
of  $\g(X_2)$ and by Lemma \ref{induced-sd} there is a solution in $\g(X_1)$ which
guarantees $F(X_1,s-s(X_2))=1$.

\medskip
Further, for every  $0\leq s \leq  s(X_2)$ where $F(X_2,s)=1$ 
we have $F(X_1\oslash X_2,s)=1$ since the solutions from $\g(X_2)$ do not contain
any predecessors of vertices from $\g(X_1)$ in $\g(X_1\oslash X_2)$. 

Also for every  $1\leq s' \leq  s(X_1)$ where $F(X_1,s')=1$ for $s=s'+s(X_2)$ we have  $F(X_1\oslash X_2,s)=1$ since
every solution  in $\g(X_1)$  has to be extended
by $X_2$ since at least one predecessor of $\g(X_2)$ is part of the solution and thus,
all vertices of $\g(X_2)$ have to belong to the  solution.

\item  
If $F(X_1  \otimes X_2,s)=1$, then we distinguish two cases. 
If the solution of size $s$ is empty, then $s=0$. 

Otherwise, $s=s(X_1)+s(X_2)$ since
$\g(X_1\otimes X_2)$ is strongly connected and thus, all vertices
of $\g(X_1)$ and all vertices of $\g(X_2)$ have to be part of  the
solution.

\medskip
Further,
if $s=0$ or $s=s(X_1)+s(X_2)$, it holds  that $F(X_1  \otimes X_2,s)=1$ since 
the empty and the complete vertex set both satisfy (\ref{dc}). 
\end{enumerate}
This shows the statements of the lemma.
\end{proof}

\begin{corollary}\label{cor1}
There is a solution with sum $s$  for an
instance of SSG such that $G$ is a directed co-graph which
is given by some  di-co-expression $X$
if and only if $F(X,s)=1$. Therefore, $OPT(I)=\max\{s \mid F(X,s)=1\}$.
\end{corollary}

\begin{theorem}\label{th-co-ssg}
SSG can be solved in directed co-graphs  with $n$
vertices and $m$ arcs in $\bigo(n\cdot c^2+m)$ time and $\bigo(n\cdot c)$ space.
\end{theorem}

\begin{proof}
Let $G=(A,E)$  be a directed co-graph and $T$ be a di-co-tree for $G$ with root $r$.
For some vertex $u$ of $T$ we denote by $T_u$
the subtree rooted at $u$ and $X_u$ the co-expression defined by $T_u$.
In order to solve the SSG problem for an instance $I$ on
graph $G$, we traverse  di-co-tree $T$ in a bottom-up order.
For every vertex $u$ of $T$ and $0\leq s \leq c$ we compute $F(X_u,s)$
following the rules given in Lemma \ref{le1}. By Corollary \ref{cor1} we can solve our
problem by $F(X_r,s)=F(X,s)$.

A di-co-tree $T$ can be computed in $\bigo(n+m)$ time from
a directed co-graph with $n$ vertices and $m$ arcs,  see \cite{CP06}.
All $s(X_i)$ can be precomputed in $\bigo(n)$ time.
Our rules given in Lemma \ref{le1} show the following running times.
\begin{itemize}
\item
For every $a_j\in A$ and every $0\leq s \leq c$ value
$F(a_j,s)$ is computable in  $\bigo(1)$ time.

\item
For every $0\leq s \leq c$, every
$F(X_1  \oplus X_2,s)$ and every $F(X_1  \oslash X_2,s)$  can be computed  in $\bigo(c)$ time from
$F(X_1,s')$ and $F(X_2,s'')$.

\item For every $0\leq s \leq c$,
every $F(X_1  \otimes X_2,s)$ can be computed
in  $\bigo(1)$ time from $s(X_1)$ and $s(X_2)$.

\end{itemize}
Since we have $n$ leaves and $n-1$ inner vertices in $T$,
the running time is in $\bigo(nc^2+m)$.
\end{proof}

\begin{example}\label{ex1-ssg-co}
We consider the SSG instance $I$ with $n=4$ items using   $\g(X)$ shown in Fig.~\ref{F02} and
defined 
by the expression in (\ref{eq-ori-c4}), $c=7$, and the following sizes.
$$
\begin{array}{l|llll}
j  &  1 & 2 & 3 & 4 \\
\hline
s_j&1 & 2 & 2  & 3
\end{array}
$$

The rules  given in  Lemma \ref{le1} lead to
the values in Table \ref{tab-ssga-co}.
Thus, the optimal solution is $\{a_2,a_3,a_4\}$ with $OPT(I)=7$.
\end{example}

\begin{table}[h!]
\caption{Table for Example \ref{ex1-ssg-co}\label{tab-ssga-co}}

$$
\begin{array}{l|llllllllll}
 &  \multicolumn{8}{c}{F(X',s)} \\
 \hline
X'~~~~~~~~~~~~~~~~~~~~~~~~~~~s          & 0 &  1 & 2 & 3 &  4 & 5 & 6 & 7  \\
\hline
v_1                                     & 1&1 & 0 & 0  & 0 & 0 & 0 & 0  \\
v_2                                     & 1&0 & 1 & 0  & 0 & 0 & 0 & 0  \\
v_3                                     & 1&0 & 1 & 0  & 0 & 0 & 0 & 0  \\
v_4                                     & 1&0 & 0 & 1  & 0 & 0 & 0 & 0  \\
\hline
v_1\oplus v_3                           & 1& 1 & 1 & 1  & 0 & 0 & 0 & 0\\
v_2 \otimes v_4                         & 1& 0 & 0 & 0  & 0 & 1 & 0 & 0\\
(v_1\oplus v_3) \oslash(v_2 \otimes v_4)& 1 &0 & 0& 0  & 0 & 1 & 1 & 1 \\
\hline
\end{array}
$$
\end{table}

A {\em tournament} is a digraph $G=(A,E)$ where for each two different vertices 
$u,v\in A$ it holds that exactly one of the two pairs  $(u,v)$ and $(v,u)$ belongs to $E$.
The class of transitive tournaments is denoted by $\TT$.
Transitive tournaments are characterized in several ways, see \cite[Chapter 9]{Gou12}.

\begin{lemma}[\cite{Gou12}]\label{s11}
For every digraph $G$ the following statements are equivalent.
\begin{enumerate}
\item \label{s11b} $G$ is a transitive tournament.
\item \label{s11c} $G$ is an acyclic tournament.
\item \label{s11e} $G$ is a tournament with exactly one Hamiltonian path.
\item \label{s11f} $G$ is a tournament and every vertex in $G$ 
has a different outdegree, i.e. $\{\odeg(v)~|~v\in V\}=\{0,\ldots,|V|-1\}$.
\item \label{s11h}  $G$ can 
be constructed from the one-vertex graph  by repeatedly adding
an  out-dominating vertex.
\item \label{s11i} $G$ 
can be constructed from the one-vertex graph repeatedly adding
an   in-dominated  vertex.
\end{enumerate}
\end{lemma}

In \cite[Lemma 4]{GMT18} it is shown that SSG is polynomial on acyclic tournaments
without stating a running time. Since acyclic tournaments, and equivalently
transitive tournaments, are a subclass of oriented co-graphs, we reconsider
the following result.

\begin{remark}\label{tt}
Every transitive tournament $G$ can be defined from a
single vertex graph $v_1$ by repeatedly adding
a vertex of maximum indegree and outdegree $0$,
i.e.\ an  in-dominated vertex $v_2,\ldots, v_n$ (cf.\ Lemma \ref{s11}). This
order can be defined in $\bigo(n^2)$ time from $G$.
The feasible solutions w.r.t. the digraph constraint  (\ref{dc})
are $\emptyset$ and
for $1\leq k \leq n$ the set $\{v_i~|~ k\leq i \leq n\}$.
This leads to at most $n+1$ possible solutions for SSG for which
we have to check the capacity
constraint  (\ref{cap}) and among those satisfying  (\ref{cap})
we select one set with largest sum of sizes.  
Thus, SSG is solvable in $\bigo(n^2)$ time on transitive tournaments with $n$ vertices.
\end{remark}

\begin{example}
We consider the SSG instance $I$ with $n=4$ items using the transitive
tournament $\g(X)$
defined by expression
\begin{equation}
X=(((v_1\oslash v_2)\oslash v_3)\oslash v_4),
\end{equation}
$c=7$, and the following sizes.
$$
\begin{array}{l|llll}
j  &  1 & 2 & 3 & 4 \\
\hline
s_j&1 & 2 & 2  & 3
\end{array}
$$
By Remark \ref{tt} for this instance the feasible
solutions of SSG w.r.t. the digraph constraint  (\ref{dc})
are $\emptyset$, $\{a_4\}$, $\{a_3,a_4\}$, $\{a_2,a_3, a_4\}$,
and  $\{a_1,a_2,a_3, a_4\}$. Among these only $\emptyset$, $\{a_4\}$, $\{a_3,a_4\}$, and
$\{a_2,a_3, a_4\}$ satisfy the capacity
constraint (\ref{cap}).
Thus, the optimal solution is $\{a_2,a_3,a_4\}$ with $OPT(I)=7$.
\end{example}

A {\em bioriented clique} is a digraph $G=(A,E)$ 
where for each two different vertices $u,v\in A$ it holds that 
both  of the two pairs  $(u,v)$ and $(v,u)$ belong to $E$.
The class of  bioriented cliques is denoted by $\BC$.

\begin{remark}\label{cl}
Within a bioriented clique $G=(A,E)$  the whole vertex set 
is a strongly connected component. By 
Lemma \ref{le-reach-sc} the only possible solutions are
$A$ and $\emptyset$. Thus, SSG is solvable in $\bigo(n)$ time on 
bioriented cliques  with $n$ vertices.
\end{remark}

\subsection{Subset sum  with  weak digraph constraint (SSGW)}

Next, we consider SSGW on directed co-graphs. 
In order to get useful informations about the sources within a
solution, we use an extended data structure.
We consider an instance of SSGW such that $G=(A,E)$ is a directed co-graph which
is given by some  di-co-expression   $X$.
For some subexpression $X'$ of $X$ let $H(X',s,s')=1$ if there is a solution
$A'$ in the graph defined by $X'$  satisfying (\ref{cap}) and (\ref{wdc})  
such that $s(A')=s$ and the sum of sizes of
the sources in $A'$ is $s'$,  otherwise let $H(X',s,s')=0$. We denote by $o(X)$
the sum of the sizes of all sources in $\g(X)$.

\begin{remark}\label{rem-w}
A remarkable difference between  SSGW and SSG w.r.t.\ co-graph operations  
is the following.
When considering $X_1\oslash X_2$ we can combine solutions $A_1$ of
$X_1$  satisfying (\ref{cap}) and (\ref{wdc}) which do not contain all items
of $X_1$ with solutions $A_2$ of $X_2$ satisfying only (\ref{cap}) to obtain solution
$A_1 \cup A_2$ of $X_1\oslash X_2$  satisfying (\ref{cap}) and (\ref{wdc}), if $s(A_1)+s(A_2)\leq c$. 
Furthermore, within $X_1\otimes X_2$ we can  combine solutions $A_1$ of
$X_1$  satisfying (\ref{cap}) which do not contain all items
and solutions $A_2$ of
$X_2$  satisfying (\ref{cap}) which do not contain all items
to obtain solution $A_1\cup A_2$
of $X_1\otimes X_2$  satisfying (\ref{cap}) and (\ref{wdc}), if $s(A_1)+s(A_2)\leq c$. 
\end{remark}

Thus, in order to solve SSGW on a directed co-graph $G$, we use 
solutions for SSP on subexpressions for $G$. 
We consider an instance of SSP such that $G=(A,E)$ is a directed co-graph which
is given by some di-co-expression $X$.
For a subexpression $X'$ of $X$ let $H'(X',s)=1$ if there is a solution $A'$ 
in the digraph defined by $X'$ satisfying (\ref{cap})
such that $s(A')=s$,  otherwise let $H'(X',s)=0$.

\begin{lemma}\label{le1a} Let $0\leq s \leq c$.
\begin{enumerate}
\item $H'(a_j,s)=1$ if and only if $s=0$ or $s=s_j$.

In all other cases  $H'(a_j,s)=0$.

\item $H'(X_1\oplus X_2,s)=1$, if  and only if 
there are  some $0\leq s'\leq s$ and $0\leq s''\leq s$
such that $s'+s''=s$ and $H'(X_1,s')=1$ and $H'(X_2,s'')=1$.

In all other cases  $H'(X_1\oplus X_2,s)=0$.

\item $H'(X_1\oslash X_2,s)=H'(X_1\oplus X_2,s)$

\item $H'(X_1\otimes X_2,s)=H'(X_1\oplus X_2,s)$
\end{enumerate}
\end{lemma}

\begin{proof}
We show the correctness of the stated equivalences. Let $0\leq s \leq c$.
\begin{enumerate}
\item 
The only possible solutions in $\g(a_j)$  are $\emptyset$ and $\{a_j\}$ which have size 
$0$ and $s_j$, respectively.

\item If $H'(X_1  \oplus X_2,s)=1$, then  by Observation \ref{induced-sd2} there are $s'$ and $s''$ 
such that $s'+s''=s$ and solutions in $\g(X_1)$  
and in $\g(X_2)$ which guarantee $H'(X_1,s')=1$ and $H'(X_2,s'')=1$.

Further, for every $s'$ and $s''$, such that $s'+s''=s$, $H'(X_1,s')=1$, and $H'(X_2,s'')=1$, we can
combine these two solutions into one solution of size $s$ in $\g(X_1\oplus X_2)$. Thus,
it holds  that $H'(X_1  \oplus X_2,s)=1$.

\item Since the arcs are irrelevant for the capacity constraint (\ref{cap}), 
it holds that $H'(X_1\oslash X_2,s)=H'(X_1\oplus X_2,s)$.

\item Since the arcs are irrelevant for the capacity constraint (\ref{cap}), 
it holds that $H'(X_1\otimes X_2,s)=H'(X_1\oplus X_2,s)$.
\end{enumerate}
This shows the statements of the lemma.
\end{proof}

This allows us to compute the values  $H(X',s,s')$ as follows.

\begin{lemma}\label{lemma-ssgw-co}
Let $0\leq s,s' \leq c$.
\begin{enumerate}
\item $H(a_j,s,s')=1$ if and only if $s=s'=0$ or $s_j=s=s'$. 

In all other cases  $H(a_j,s,s')=0$.

\item 
$H(X_1\oplus X_2,s,s')=1$, if  and only if 
there are  $0\leq s_1\leq s$, $0\leq s_2\leq s$, $0\leq s'_1\leq s'$, $0\leq s'_2\leq s'$,
such that $s_1+s_2=s$, $s'_1+s'_2=s'$, $H(X_1,s_1,s'_1)=1$,
and $H(X_2,s_2,s'_2)=1$.

In all other cases  $H(X_1\oplus X_2,s,s')=0$.

\item $H(X_1\oslash X_2,s,s')=1$, if  and only if 
\begin{itemize}
\item
$H(X_1,s,s')=1$ for $1\leq s < s(X_1)$ or

\item
$H'(X_2,s)=1$ for $0\leq s \leq  s(X_2)$\footnote{The value $s=0$ is for 
choosing an empty solution in $\g(X_1\oslash X_2)$.} and $s'=0$ or

\item
there are   $1\leq s_2 \leq s(X_2)$, 
such that $s(X_1)+s_2=s$, $o(X_1)=s'$, and $H(X_2,s_2,o(X_2))=1$, or

\item
$s=s(X_1)+s(X_2)$ and $s'=o(X_1)$, or
\item
there are $0\leq s_1 < s(X_1)$, $0\leq s_2\leq s(X_2)$, 
such that $s_1+s_2=s$,  $H(X_1,s_1,s')=1$,
and $H'(X_2,s_2)=1$.

\end{itemize}
In all other cases  $H(X_1\oslash X_2,s,s')=0$.

\item $H(X_1\otimes X_2,s,0)=1$, if  and only if 
\begin{itemize}
\item 
$H'(X_1,s)=1$ for $1\leq s < s(X_1)$ or

\item
$H'(X_2,s)=1$ for $0\leq s <  s(X_2)$\footnote{The value $s=0$ is for 
choosing an empty solution in $\g(X_1\otimes X_2)$.}  or
\item

there are   $1\leq s_2 \leq s(X_2)$, 
such that $s(X_1)+s_2=s$,  and $H(X_2,s_2,o(X_2))=1$, or

\item
there are   $1\leq s_1 \leq s(X_1)$, 
such that $s_1+s(X_2)=s$, and $H(X_1,s_1,o(X_1))=1$, or

\item
$s=s(X_1)+s(X_2)$, or

\item
there exist  $1\leq s_1< s(X_1)$ and $1\leq s_2< s(X_2)$
such that $s_1+s_2=s$, $H'(X_1,s_1)=1$, and $H'(X_2,s_2)=1$.
\end{itemize}

In all other cases  $H(X_1\otimes X_2,s,s')=0$.
\end{enumerate}
\end{lemma}

\begin{proof}
We show the correctness of the stated equivalences. Let $0\leq s,s' \leq c$.
\begin{enumerate}
\item 
The only possible solutions in $\g(a_j)$  are $\emptyset$ and $\{a_j\}$ which have size 
$0$ and $s_j$, respectively. Further, a single vertex is a source.

\item 
If $H(X_1  \oplus X_2,s,s')=1$, then  by Lemma \ref{induced-sd3} there are $s_1$, $s_2$ and $s'_1$, $s'_2$ 
such that $s_1+s_2=s$, $s'_1+s'_2=s'$ and solutions in $\g(X_1)$  
and in $\g(X_2)$ which guarantee $H(X_1,s_1,s'_1)=1$ and $H(X_2,s_,s'_2)=1$.

\medskip
Further, for every $0\leq s_1\leq s$, $0\leq s_2\leq s$, $0\leq s'_1\leq s'$, $0\leq s'_2\leq s'$,
such that $s_1+s_2=s$, $s'_1+s'_2=s'$, $H(X_1,s_1,s'_1)=1$,
and $H(X_2,s_2,s'_2)=1$, it holds  that $H(X_1  \oplus X_2,s,s')=1$
since the operation (disjoint union) 
does not create  new edges.

\item 
If $H(X_1  \oslash X_2,s,s')=1$,  then we distinguish the following cases.
If the solution of size $s$ in $\g(X_1\oslash X_2)$ is a non-empty proper 
subset of the vertices of  $\g(X_1)$, then by Lemma \ref{induced-sd4} there 
is a solution in $\g(X_1)$ which guarantees $H(X_1,s,s')=1$.

Next, assume that the solution $A'$ of size $s$ in $\g(X_1\oslash X_2)$ contains no 
vertices of  $\g(X_1)$. Since every solution  satisfying constraints (\ref{cap}) and
(\ref{wdc}) is also a solution which satisfies only  (\ref{cap}), we have
$H'(X_1  \oslash X_2,s)=1$. And since $A'$  contains no
vertices of  $\g(X_1)$, Observation \ref{induced-sd2} implies that 
there is a solution in $\g(X_2)$ which
guarantees $H'(X_2,s)=1$.

If the solution $A'$ of size $s$ in $\g(X_1\oslash X_2)$ contains  all
vertices of  $\g(X_1)$, then the order
composition and the weak digraph constraint (\ref{wdc}) imply 
that the set $A'$ can be extended by every solution of $\g(X_2)$
which includes all sources of $\g(X_2)$. Thus, by Lemma \ref{induced-sd5}, there
is a solution in $\g(X_2)$, which guarantees  $H(X_2,s-s(X_1),o(X_2))=1$.

Further, if the solution $A'$ of size $s$ in $\g(X_1\oslash X_2)$ contains  all
vertices of  $\g(X_1)$, it is also possible to extend $A'$ by  all
vertices of  $\g(X_2)$ and thus  $s=s(X_1)+s(X_2)$.

Finally, if the solution $A'$ of size $s$ in $\g(X_1\oslash X_2)$ contains  
some but not all vertices of  $\g(X_1)$ and possibly vertices of $\g(X_2)$, then by  Lemma
\ref{induced-sd4} and Observation \ref{induced-sd2}  there are  
$s_1$ and $s_2$ such that $s_1+s_2=s$ and solutions in $\g(X_1)$  
and in $\g(X_2)$ which guarantee $H(X_1,s_1,s')=1$ and $H'(X_2,s_2)=1$.

\medskip
The solutions of size $1\leq s < s(X_1)$ from $\g(X_1)$ 
remain feasible for $\g(X_1\oslash X_2)$. 

Every subset $A'$ of size $0\leq s \leq  s(X_2)$ from $\g(X_2)$ which satisfies  (\ref{cap}) 
leads to a solution $A'$ of size $s$ satisfying (\ref{cap}) and
(\ref{wdc}) in $\g(X_1\oslash X_2)$ since every vertex of  $\g(X_2)$
gets a predecessor in  $\g(X_1)$, which is not in $A'$.

Further, the set of all vertices of $\g(X_1)$  extended by every solution of $\g(X_2)$ of size $s_2$
which includes all sources of $\g(X_2)$ leads to a feasible solution for $\g(X_1\oslash X_2)$
of size $s(X_1)+s_2$.  
The size of the sources has to be updated to $o(X_1)$, since the sources of $\g(X_1)$ are 
the sources of $\g(X_1\oslash X_2)$.

Moreover, the complete vertex set of $\g(X_1\oslash X_2)$ is obviously
a feasible SSGW solution if it fulfills the capacity constraint.

Furthermore, by Remark \ref{rem-w} we can combine SSGW solutions 
of size $s_1< s(X_1)$ of $\g(X_1)$ and  SSP solutions  of size $s_2$ of $\g(X_2)$ to a SSGW solution
of size $s_1+s_2$ of $\g(X_1\oslash X_2)$.

\item 
First, we want to mention that $H(X_1  \otimes X_2,s,s')=1$ is only possible 
for $s'=0$, since $\g(X_1  \otimes X_2)$
has no sources.
If $H(X_1  \otimes X_2,s,0)=1$, we distinguish the following cases.
Assume the solution of size $s$ in $\g(X_1\otimes X_2)$ is a proper and non-empty  subset of the
vertices of  $\g(X_1)$. Since $H(X_1  \otimes X_2,s,0)=1$, it holds that $H'(X_1  \otimes X_2,s)=1$. 
And since $A'$  contains only 
vertices of  $\g(X_1)$, Observation  \ref{induced-sd2} implies that 
there is a solution in $\g(X_1)$ which
guarantees $H'(X_1,s)=1$.

If the solution of size $s$ in $\g(X_1\otimes X_2)$  is a proper subset
of the vertices of  $\g(X_2)$, then by the same arguments as for $\g(X_1)$  there 
is a solution in $\g(X_2)$ which guarantees $H'(X_2,s)=1$.

If the solution $A'$ of size $s$ in $\g(X_1\otimes X_2)$ contains  all
vertices of  $\g(X_1)$, then the series
composition and the weak digraph constraint (\ref{wdc}) imply 
that the set $A'$ can be extended by all solutions of $\g(X_2)$
which include all sources of $\g(X_2)$. Thus, by Lemma \ref{induced-sd5}, there
is a solution in $\g(X_2)$, which guarantees  $H(X_2,s-s(X_1),o(X_2))$.

If the solution $A'$ of size $s$ in $\g(X_1\otimes X_2)$ contains  all
vertices of  $\g(X_2)$,  then by the same arguments as for $\g(X_1)$ 
there
is a solution in $\g(X_1)$, which guarantees  $H(X_1,s-s(X_2),o(X_1))$.

If the solution $A'$ of size $s$ in $\g(X_1\otimes X_2)$ contains  all
vertices of  $\g(X_1)$ or all
vertices of  $\g(X_2)$, then by (\ref{wdc}) solution $A'$ can be extended by
all vertices of  $\g(X_2)$  or  all
vertices of  $\g(X_1)$, respectively, and thus  $s=s(X_1)+s(X_2)$.

Finally, if the solution $A'$ of size $s$ in $\g(X_1\otimes X_2)$ contains  
some but not all vertices of  $\g(X_1)$ and some but not all vertices of $\g(X_2)$, then by
Observation \ref{induced-sd2} there are  
$s_1$ and $s_2$ such that $s_1+s_2=s$ and solutions in $\g(X_1)$  
and in $\g(X_2)$ which guarantee $H'(X_1,s_1)=1$ and $H'(X_2,s_2)=1$.

\medskip
Every subset $A'$ of size $1\leq s <  s(X_1)$ from $\g(X_1)$ which satisfies  (\ref{cap}) 
leads to a solution $A'$ of size $s$ satisfying constraints (\ref{cap}) and
(\ref{wdc}) in $\g(X_1\otimes X_2)$ since every vertex of  $\g(X_1)$
gets a predecessor in  $\g(X_2)$, which is not in $A'$.

In the same way every subset $A'$ of size $0\leq s <  s(X_2)$ from $\g(X_2)$ which satisfies  (\ref{cap}) 
leads to a solution $A'$ of size $s$ satisfying (\ref{cap}) and
(\ref{wdc}) in $\g(X_1\otimes X_2)$ since every vertex of  $\g(X_2)$
gets a predecessor in  $\g(X_1)$, which is not in $A'$.

Further, all the set of all vertices of $\g(X_1)$  extended by every solution of $\g(X_2)$ of size $s_2$
which includes all sources of $\g(X_2)$ leads to a feasible solution for $\g(X_1\otimes X_2)$
of size $s(X_1)+s_2$ 
and the set of all vertices of $\g(X_2)$  extended by every solution of $\g(X_1)$ of size $s_1$
which includes all sources of $\g(X_1)$ leads to a feasible solution for $\g(X_1\otimes X_2)$ of size $s_1+s(X_2)$.

Moreover, the complete vertex set of $\g(X_1\otimes X_2)$ is obviously
a feasible SSGW solution.

Furthermore, by Remark \ref{rem-w}, we can combine SSP solutions of
size $s_1< s(X_1)$ and SSP solutions of size $s_2<s(X_2)$ to a SSGW solution
of size $s_1+s_2$.
\end{enumerate}
This shows the statements of the lemma.
\end{proof}

In order to solve the SSGW problem we traverse  di-co-tree $T$ in a bottom-up order
and perform the following computations depending on the type of operation.

\begin{corollary}
There is a solution with sum $s$  for some
instance of SSGW such that $G$ is a directed co-graph which
is given by some  di-co-expression $X$
if and only if $H(X,s,s')=1$. Therefore, $OPT(I)=\max\{s \mid H(X,s,s')=1\}$.
\end{corollary}

The next result can be obtained by similar arguments as given within
the proof of Theorem \ref{th-co-ssg}.

\begin{theorem}\label{ssgw}
SSGW can be solved in directed co-graphs with $n$
vertices and $m$ arcs
in  $\bigo(n\cdot c^4 +m)$ time and  $\bigo(n\cdot c^2)$ space.
\end{theorem}

\begin{example}\label{ex-ssgw-co}
We consider the SSGW instance $I$ with $n=4$ items using   $\g(X)$ shown in Fig.~\ref{F02} and
defined 
by the expression in 
(\ref{eq-ori-c4}), $c=7$, and the following sizes.
$$
\begin{array}{l|llll}
j  &  1 & 2 & 3 & 4 \\
\hline
s_j&1 & 2 & 2  & 3
\end{array}
$$

The rules  given in  Lemma \ref{le1a} lead to the values in Table \ref{tab-ssgx-co}  and the rules given in Lemma \ref{lemma-ssgw-co} lead to
the values in Table \ref{tab-bigg-a}.
Thus, the optimal solution is $\{a_2,a_3,a_4\}$ with $OPT(I)=7$.
\end{example}

\begin{table}[h!]
\caption{Table for Example \ref{ex-ssgw-co}\label{tab-ssgx-co}}
$$
\begin{array}{l|llllllll}
                                        &\multicolumn{8}{c}{H'(X',s)}      \\
                                \hline
X'~~~~~~~~~~~~~~~~~~~~~~~~~~~s          &0&1&2&3&4&5&6&7 \\
\hline
v_1                                     &1&1&0&0&0&0&0&0  \\
v_2                                     &1&0&1&0&0&0&0&0  \\
v_3                                     &1&0&1&0&0&0&0&0  \\
v_4                                     &1&0&0&1&0&0&0&0  \\
\hline
v_1\oplus v_3                           &1&1&1&1&0&0&0&0  \\
v_2 \otimes v_4                         &1&0&1&1&0&1&0&0  \\
(v_1\oplus v_3)\oslash(v_2 \otimes v_4) &1&1&1&1&1&1&1&1  \\
\hline
\end{array}
$$
\end{table}

\begin{table}[h!]
\centering
\caption{Table for Example \ref{ex-ssg-spg}\label{tab-bigg-a}}

\rotatebox{90}{\footnotesize
$
\begin{array}{l|llllllll|llllllll|llllllll|llllllll|}
  & \multicolumn{32}{c|}{H(X',s,s')}   \\
\hline
      & \multicolumn{8}{c|}{s'=0}   & \multicolumn{8}{c|}{s'=1}     & \multicolumn{8}{c|}{s'=2}   & \multicolumn{8}{c|}{s'=3}\\
X'~~~~~~~~~~~~~~~~~~~~~~~~~~~~~~~~~~~~~~~~~~~~~~~s              &0&1&2&3&4&5&6&7  & 0& 1 & 2 & 3 &  4 & 5 & 6 & 7   & 0&  1 & 2 & 3 &  4 & 5 & 6 & 7 & 0& 1 & 2 & 3 &  4 & 5 & 6 & 7   \\
\hline
v_1                                                             &1 &0&0&0&0&0&0&0& 0&1&0&0&0&0&0&0& 0&0&0&0&0&0&0&0& 0&0&0&0&0&0&0&0\\
v_2                                                             &1 &0&0&0&0&0&0&0& 0&0&0&0&0&0&0&0& 0&0&1&0&0&0&0&0& 0&0&0&0&0&0&0&0\\
v_3                                                             &1 &0&0&0&0&0&0&0& 0&0&0&0&0&0&0&0& 0&0&1&0&0&0&0&0& 0&0&0&0&0&0&0&0\\
v_4                                                             &1 &0&0&0&0&0&0&0& 0&0&0&0&0&0&0&0& 0&0&0&0&0&0&0&0& 0&0&0&1&0&0&0&0\\

\hline
v_1\oplus v_3                                                & 1&0&0&0&0&0&0&0& 0&1&0&0&0&0&0&0& 0&0&1&0&0&0&0&0& 0&0&0&1&0&0&0&0\\
v_2 \otimes v_4                                              & 1&0&0&0&0&1&0&0& 0&0&0&0&0&0&0&0& 0&0&0&0&0&0&0&0& 0&0&0&0&0&0&0&0\\
(v_1\oplus v_3)\oslash(v_2 \otimes v_4)                      & 1&0&1&1&0&1&0&0& 0&1&0&1&1&0&1&0& 0&0&1&0&1&1&0&1& 0&0&0&1&0&0&0&0\\

\hline
\\
\\
\\
  & \multicolumn{32}{c|}{H(X',s,s')}   \\
\hline
      &  \multicolumn{8}{c|}{s'=4}     & \multicolumn{8}{c|}{s'=5}   & \multicolumn{8}{c|}{s'=6}  & \multicolumn{8}{c|}{s'=7} \\
X'~~~~~~~~~~~~~~~~~~~~~~~~~~~~~~~~~~~~~~~~~~~~~~~s              &0&1&2&3&4&5&6&7  & 0& 1 & 2 & 3 &  4 & 5 & 6 & 7   & 0&  1 & 2 & 3 &  4 & 5 & 6 & 7 & 0& 1 & 2 & 3 &  4 & 5 & 6 & 7   \\
\hline
v_1                                                             &0&0&0&0&0&0&0&0& 0&0&0&0&0&0&0&0& 0&0&0&0&0&0&0&0& 0&0&0&0&0&0&0&0\\
v_2                                                             &0&0&0&0&0&0&0&0& 0&0&0&0&0&0&0&0& 0&0&0&0&0&0&0&0& 0&0&0&0&0&0&0&0\\
v_3                                                             &0&0&0&0&0&0&0&0& 0&0&0&0&0&0&0&0& 0&0&0&0&0&0&0&0& 0&0&0&0&0&0&0&0\\
v_4                                                             &0&0&0&0&0&0&0&0& 0&0&0&0&0&0&0&0& 0&0&0&0&0&0&0&0& 0&0&0&0&0&0&0&0\\

\hline
v_1\oplus v_3                                               & 0&0&0&0&0&0&0&0& 0&0&0&0&0&0&0&0& 0&0&0&0&0&0&0&0& 0&0&0&0&0&0&0&0\\
v_2 \otimes v_4                                              & 0&0&0&0&0&0&0&0& 0&0&0&0&0&0&0&0& 0&0&0&0&0&0&0&0& 0&0&0&0&0&0&0&0\\
(v_1\oplus v_3)\oslash(v_2 \otimes v_4)                      & 0&0&0&0&0&0&0&0& 0&0&0&0&0&0&0&0& 0&0&0&0&0&0&0&0& 0&0&0&0&0&0&0&0\\

\hline

\end{array}
$
}
\end{table}

\section{SSG and SSGW on series-parallel digraphs}\label{sol-spd}

\subsection{Series-parallel digraphs}

We recall the definitions from \cite{BG18} which are
based on \cite{VTL82}. First, we introduce two operations
for two vertex-disjoint digraphs $G_1=(V_1,E_1)$ and $G_2=(V_2,E_2)$.
Let $O_1$ be the set of vertices of outdegree $0$ (set of sinks) in $G_1$ and
$I_2$ be the set of vertices of indegree $0$ (set of sources) in $G_2$.

\begin{itemize}
\item 
The {\em parallel composition} of $G_1$ and $G_2$,
denoted by $G_1 \cup G_2$, is the digraph with vertex set $V_1\cup V_2$
and arc set $E_1\cup E_2$.

\item 
The {\em  series composition}  of $G_1$ and $G_2$,
denoted by $G_1 \times G_2$ is the digraph with vertex set $V_1\cup V_2$
and arc set $E_1\cup E_2\cup (O_1 \times I_2)$.
\end{itemize}

\begin{definition}[Minimal series-parallel digraphs]
The class of {\em minimal series-parallel digraphs}, {\em msp-digraphs} for
short, is recursively defined as follows.
\begin{enumerate}
\item Every digraph on a single vertex $(\{v\},\emptyset)$,
denoted by $v$, is a {\em minimal series-parallel digraph}.

\item If  $G_1$ and $G_2$
are vertex-disjoint minimal series-parallel digraphs, then
\begin{enumerate}
\item the parallel composition
$G_1 \cup G_2$ and

\item  then series composition
$G_1 \times G_2$
are  {\em minimal series-parallel digraphs}.
\end{enumerate}
\end{enumerate}
The class of minimal series-parallel digraphs is denoted by $\MSP$.
\end{definition}

Every expression $X$ using  these three operations
is called an {\em msp-expression} and
$\g(X)$ the defined digraph. 

\begin{example}\label{ex-msp}
\begin{enumerate}
\item
The msp-expression 
\begin{equation}
X=((v_1\cup v_2) \times (v_3 \cup v_4)) \label{eq-ori-c4msp}
\end{equation}
defines  $\g(X)$ shown in Fig.~\ref{F04}.

\item
The msp-expression 
\begin{equation}
X=(((v_1\times v_2) \cup (v_3\times v_4)) \times (v_5 \times v_6)) \label{eq-ori-c4msp2}
\end{equation}
defines  $\g(X)$ shown in Fig.~\ref{F05}.
\end{enumerate}
\end{example}

\begin{figure}[hbtp]
\centering
\parbox[b]{60mm}{
\centerline{\epsfxsize=30mm \epsfbox{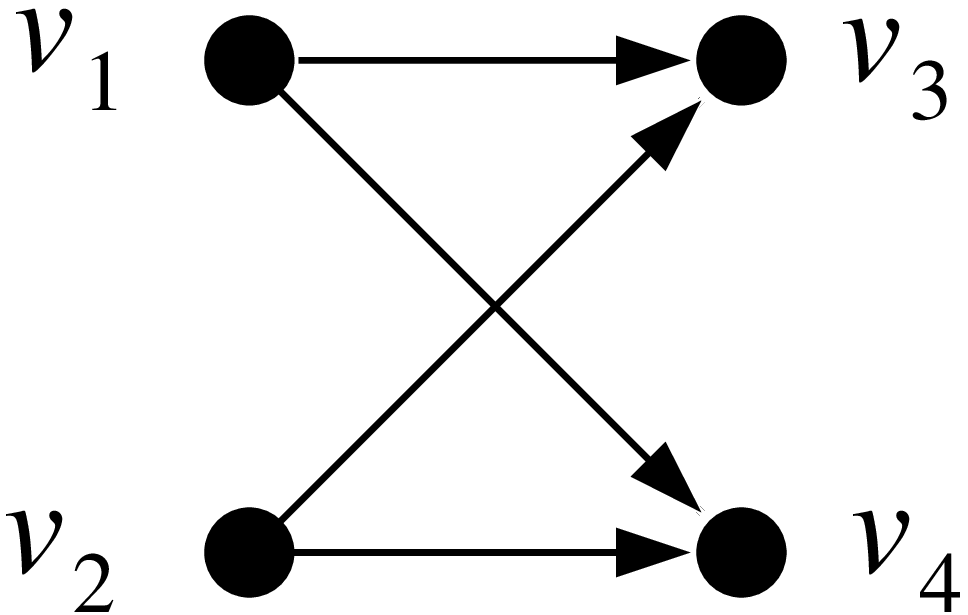}}
\caption{Digraph  in Example \ref{ex-msp}(1.).}
\label{F04}}
\hspace{1cm}
\parbox[b]{60mm}{
\centerline{\epsfxsize=45mm \epsfbox{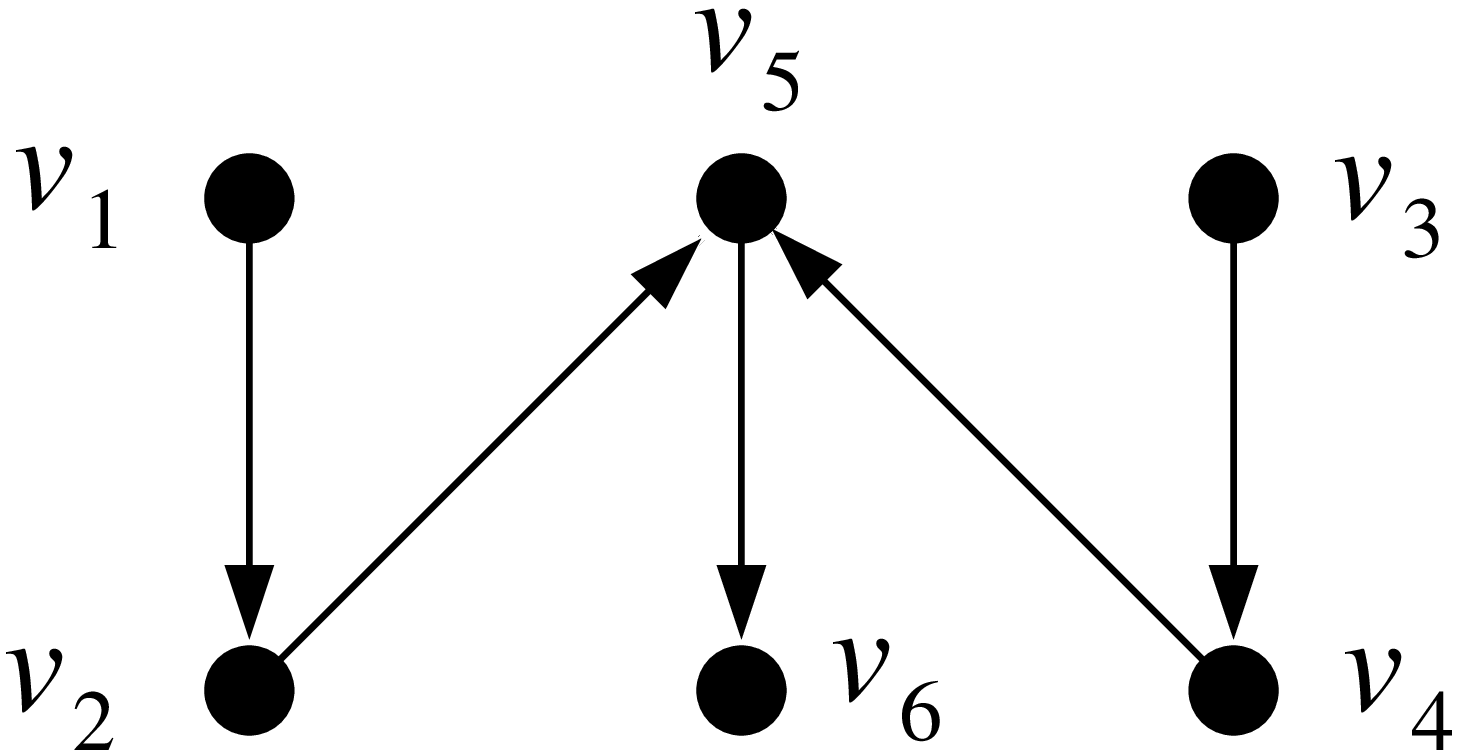}}
\caption{Digraph in Example \ref{ex-msp}(2.).
}
\label{F05}}
\end{figure}

For every minimal series-parallel digraph we can define a tree structure,
denoted as {\em msp-tree}.\footnote{In \cite{VTL82} the tree-structure for an msp-digraphs
is denoted as binary decomposition tree.} The leaves of the msp-tree represent the
vertices of the graph and the inner vertices of the msp-tree  correspond
to the operations applied on the subexpressions defined by the subtrees.
For every minimal series-parallel digraph one can construct a msp-tree in linear time,
see \cite{VTL82}.

\medskip
\begin{observation}\label{ex-outtsp}
Every in- or out-rooted
tree is a  minimal series-parallel digraph.
\end{observation}

\begin{lemma}\label{lemma-reach}
Let $G=(V,E)$ be a minimal series-parallel digraph. Then, for every vertex $x\in V$
there is a sink $x_s$ of $G$, such that there is a directed path from $x$ to $x_s$ in $G$ and there
is a source $x_o$ of $G$, such that there is a path from $x_o$ to $x$ in $G$.
\end{lemma}

\begin{proof}
Can be shown by induction on the recursive definition of minimal 
series-parallel digraphs. 
\end{proof}

\begin{lemma}\label{lemma-sink}
Let $G=(V,E)$ be a minimal series-parallel digraph. Then, every non-empty
feasible solution of SSG  contains a sink of $G$.
\end{lemma}

\begin{proof}
If a feasible solution $A'$ contains some $x\in V$, then Lemma \ref{lemma-reach}
implies that there is a
sink $x_s$ of $G$, such that there is a path from $x$ to $x_s$ in $G$
which implies by (\ref{dc}) that $x_s\in A'$.
\end{proof}

\begin{definition}[Series-parallel digraphs]
Series-parallel digraphs are exactly the digraphs whose transitive closure equals
the transitive closure of some minimal series-parallel digraph.

The class of series-parallel digraphs is denoted by $\SPD$.
\end{definition}

\begin{theorem}[\cite{VTL82}]\label{th-sp}
An acyclic digraph is series-parallel, if and
only if its transitive closure is $N$-free, 
where $N=(\{u,v,w,x\},\{(v,w),(u,w),(u,x)\})$.
\end{theorem}

In order to define series-parallel partial order digraphs
by  series-parallel partial orders, we introduce two operations.
Let $(X_1 ,\leq)$ and $(X_2,\leq)$ be two partially ordered sets over a set $X$, such 
that $X_1\subseteq X$, $X_2\subseteq X$, and  $X_1 \cap X_2 =\emptyset$. 
\begin{itemize}
\item The {\em series composition} of $(X_1 ,\leq)$ and $(X_2,\leq)$
is the order with the following properties. If $x$ and $y$ 
are of the same set, then their order does not change. If $x\in X_1$
and $y\in X_2$, then it holds that  $x \leq  y$.
\item
The {\em parallel composition} of $(X_1 ,\leq)$ and $(X_2,\leq)$
is the order with the following properties.
Elements $x$ and $y$ are comparable if and only if they
are both comparable in $X_1$ or both comparable in $X_2$ and they 
keep their corresponding order.
\end{itemize}

\begin{definition}[Series-parallel partial order] The class of
{\em series-parallel partial orders} over a set $X$ 
is recursively defined as follows. 
\begin{enumerate}
\item Every single element $(\{x\},\emptyset)$, $x\in X$,   is a {\em series-parallel partial order}.
\item If $(X_1 ,\leq)$ and $(X_2,\leq)$  are series-parallel partial orders over set $X$,  such 
that $X_1\subseteq X$, $X_2\subseteq X$, and  $X_1 \cap X_2 =\emptyset$, then
\begin{enumerate}
\item 
the series composition of $(X_1 ,\leq)$ and $(X_2,\leq)$ and
\item
the parallel composition of $(X_1 ,\leq)$ and $(X_2,\leq)$ are {\em series-parallel partial orders}.
\end{enumerate}
\end{enumerate}
\end{definition}

\begin{example}\label{ex-pos}
The following  partially ordered sets are series-parallel partial orders over set $\{x_1,x_2,x_3,x_4\}$.
\begin{itemize}
\item The parallel composition of $(\{x_1\},\emptyset)$ and $(\{x_3\},\emptyset)$
leads to the series-parallel partial order $(\{x_1,x_3\},\emptyset)$.
\item The series composition of $(\{x_2\},\emptyset)$ and $(\{x_4\},\emptyset)$
leads to the series-parallel partial order  $(\{x_2,x_4\},\{(x_2,x_4)\})$.
\item The series composition of $(\{x_1,x_3\},\emptyset)$ and  $(\{x_2,x_4\},\{(x_2,x_4)\})$
leads to the series-parallel partial order  $(\{x_1,x_2,x_3,x_4\},\{(x_2,x_4),(x_1,x_2),(x_1,x_4),(x_3,x_2),(x_3,x_4)\})$.
\end{itemize}
\end{example}

\begin{definition}[Series-parallel partial order digraphs]
A {\em series-parallel partial order  digraph} $G=(V,E)$ is a digraph, where $(V,\leq)$
is a series-parallel partial order  and $(x,y)\in E$ if and only if
$x \neq y$ and $x \leq  y$.

The class of series-parallel partial order digraphs is denoted by $\SPO$.
\end{definition}

\begin{example}\label{ex-spo} The series-parallel partial orders given in
Example \ref{ex-pos} show that the digraph shown in Fig.~\ref{F03}
is a series-parallel partial order  digraph.
\end{example}

Comparing the definitions of the order composition of oriented
co-graphs with the series composition of series-parallel partial order  digraphs
and the disjoint union composition of oriented
co-graphs  with the parallel composition of series-parallel partial order  digraphs, 
see Examples  \ref{ex-orico} and \ref{ex-pos}, we obtain the following result.

\medskip
\begin{observation}
The sets $\OC$ and $\SPO$ are equal.
\end{observation}

\medskip

In Fig.~\ref{grcl} we summarize the relation of directed co-graphs, series-parallel digraphs
and related graph classes. 
The directed edges represent the existing relations between
the graph classes, which follow by their definitions. 
For the relations to further graph classes we refer to  \cite[Figure 11.1]{BG18}.

\begin{figure}
\begin{center}
\epsfig{figure=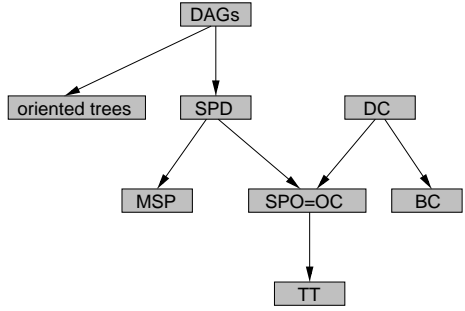,width=6.0cm}
\caption{The figure shows the inclusions of special graph classes. A directed edge
from class $A$ to class $B$ indicates that $B\subseteq A$. Two classes $A$ and $B$ 
are incomparable, if there is neither a directed path from $A$ to $B$, nor
a directed path from $B$ to $A$.}\label{grcl}
\end{center}
\end{figure}

Since SSP corresponds to SSG and also to SSGW on a digraph without arcs,
which is a minimal series-parallel digraph, we obtain the following result.

\begin{proposition}
SSG and SSGW are NP-hard on  minimal series-parallel digraph.
\end{proposition}

Next, we will show pseudo-polynomial solutions for SSG and SSGW restricted to
(minimal) series-parallel digraphs. The main idea is a dynamic programming along 
the recursive structure of a given (minimal) series-parallel digraph.

\subsection{Subset sum  with  digraph constraint (SSG)}

We consider
an instance of SSG such that $G=(A,E)$ is a minimal series-parallel digraph which
is given by some msp-expression $X$.
For some subexpression $X'$ of $X$ let $F(X',s)=1$ if there is a solution
$A'$ in the graph defined by $X'$  satisfying (\ref{cap}) and (\ref{dc}) 
such that $s(A')=s$, otherwise let $F(X',s)=0$.
We use the notation $s(X')= \sum_{a_j\in X'}s_j$.

\begin{lemma}\label{le1sp} Let $0\leq s \leq c$.
\begin{enumerate}
\item $F(a_j,s)=1$ if and only if  $s=0$ or $s_j=s$.

In all other cases  $F(a_j,s)=0$.

\item $F(X_1\cup X_2,s)=1$, if  and only if 
there are  some $0\leq s'\leq s$ and $0\leq s''\leq s$
such that $s'+s''=s$ and $F(X_1,s')=1$ and $F(X_2,s'')=1$.

In all other cases  $F(X_1\cup X_2,s)=0$.

\item $F(X_1\times X_2,s)=1$, if  and only if 
\begin{itemize}
\item
$F(X_2,s)=1$ for $0\leq s \leq  s(X_2)$\footnote{The value $s=0$ is for 
choosing an empty solution in $\g(X_1\times X_2)$.} or

\item
there is some $1\leq s'\leq s(X_1)$ such that $s=s'+s(X_2)$ and  $F(X_1,s')=1$.
\end{itemize}
In all other cases  $F(X_1\times X_2,s)=0$.
\end{enumerate}
\end{lemma}

\begin{proof}
We show the correctness of the stated equivalences. Let $0\leq s \leq c$.
\begin{enumerate}
\item 
The only possible solutions in $\g(a_j)$  are $\emptyset$ and $\{a_j\}$ which have size 
$0$ and $s_j$, respectively.

\item If $F(X_1  \cup X_2,s)=1$, then  by Lemma \ref{induced-sd} there are 
$s'$ and $s''$ such that $s'+s''=s$ and solutions in $\g(X_1)$  
and in $\g(X_2)$ which guarantee $F(X_1,s')=1$ and $F(X_2,s'')=1$.

\medskip
Further,
for every $s'$ and $s''$, such that $s'+s''=s$, $F(X_1,s')=1$, and $F(X_2,s'')=1$,
it holds  that $F(X_1  \cup X_2,s)=1$ since the parallel composition creates no additional arcs.

\item 
If $F(X_1  \times X_2,s)=1$, then we distinguish two cases.
If the solution of size $s$ in $\g(X_1\times X_2)$ contains no vertex
of  $\g(X_1)$, then by Lemma \ref{induced-sd} there is a solution in $\g(X_2)$ which
guarantees $F(X_2,s)=1$.  

Otherwise, the solution $A'$ of size $s$ in $\g(X_1\times X_2)$ contains at least one vertex
of  $\g(X_1)$.
By the definition of the series  composition and the digraph constraint
(\ref{dc}) every solution  from $\g(X_1)$ which contains a sink has to be extended
by every vertex of $X_2$ which is reachable by  a source from $\g(X_2)$.
Since by Lemma \ref{lemma-sink} 
every non-empty
feasible solution of SSG  contains a sink,
every solution  from $\g(X_1)$ has to be extended
by every vertex of $X_2$ which reachable by  a source from $\g(X_2)$.
By Lemma \ref{lemma-reach} every solution  from $\g(X_1)$ has to be extended by all
vertices of $\g(X_2)$.
Thus, by Lemma \ref{induced-sd} there is a solution in $\g(X_1)$ which
guarantees $F(X_1,s-s(X_2))=1$.

\medskip
Further, for every  $0\leq s \leq  s(X_2)$ where $F(X_2,s)=1$ 
we have $F(X_1\times X_2,s)=1$ since the solutions from $\g(X_2)$ do not contain
any predecessors of vertices from $\g(X_1)$ in $\g(X_1\times X_2)$. 

For every  $1\leq s' \leq  s(X_1)$ where $F(X_1,s')=1$  
the definition of the series  composition and the digraph constraint
(\ref{dc}) imply that for $s=s'+s(X_2)$ it holds that  $F(X_1\times X_2,s)=1$ for  reasons
given above.
\end{enumerate}
This shows the statements of the lemma.
\end{proof}

\begin{corollary}\label{cor1-sp}
There is a solution with sum $s$  for some
instance of SSG such that $G$ is a minimal series-parallel digraph which
is given by some  msp-expression $X$
if and only if $F(X,s)=1$. Therefore, $OPT(I)=\max\{s \mid F(X,s)=1\}$.
\end{corollary}

\begin{theorem}\label{th-co-ssg-msp}
SSG can be solved in minimal series-parallel digraphs with $n$
vertices and $m$ arcs
in $\bigo(n\cdot c^2+m)$ time and $\bigo(n\cdot c)$ space.
\end{theorem}

\begin{proof}
Let $G=(V,E)$  be a minimal series-parallel digraph and $T$ be an
msp-tree for $G$ with root $r$.
For some vertex $u$ of $T$ we denote by $T_u$
the subtree rooted at $u$ and $X_u$ the msp-expression defined by $T_u$.
In order to solve the SSG problem for an instance $I$
graph $G$, we traverse  msp-tree $T$ in a bottom-up order.
For every vertex $u$ of $T$ and $0\leq s \leq c$ we compute $F(X_u,s)$
following the rules given in Lemma \ref{le1sp}. By Corollary \ref{cor1-sp} we can solve our
problem by $F(X_r,s)=F(X,s)$.

An msp-tree $T$ can be computed in $\bigo(n+m)$ time from
a  minimal series-parallel digraph with $n$ vertices and $m$ arcs,  see \cite{VTL82}.
All $s(X_i)$ can be precomputed in $\bigo(n)$ time.
Our rules given in Lemma \ref{le1sp} show the following running times.
\begin{itemize}
\item
For every $a_j\in V$ and every $0\leq s \leq c$ value
$F(a_j,s)$ is computable in $\bigo(1)$ time.

\item
For every $0\leq s \leq c$, every
$F(X_1  \cup X_2,s)$   can be computed in  $\bigo(c)$ time from
$F(X_1,s')$ and $F(X_2,s'')$.

\item For every $0\leq s \leq c$,
every $F(X_1  \times X_2,s)$ can be computed
in   $\bigo(1)$ time from $F(X_1,s')$, $F(X_2,s'')$,
and $s(X_2)$.
\end{itemize}
Since we have $n$ leaves and $n-1$ inner vertices in $T$,
the running time is in $\bigo(nc^2+m)$.
\end{proof}

\begin{example}\label{ex-ssgx-spg}
We consider the SSG instance $I$ with $n=6$ items using  $\g(X)$
defined 
by the expression in 
(\ref{eq-ori-c4msp2}), $c=7$, and the following sizes.
$$
\begin{array}{l|llllll}
j  &  1 & 2 & 3 & 4 & 5 & 6\\
\hline
s_j&  2 & 1 & 4  & 3 & 2 & 3
\end{array}
$$

The rules  given in  Lemma \ref{le1sp} lead to
the values in Table \ref{tab-ssgx-spg}.
Thus, the optimal solution is $\{a_2,a_5,a_6\}$ with $OPT(I)=6$.
\end{example}

\begin{table}[h!]
\caption{Table for Example \ref{ex-ssgx-spg}\label{tab-ssgx-spg}}
$$
\begin{array}{l|lllllllll}
 &  \multicolumn{8}{c}{F(X',s)} \\
 \hline
X'~~~~~~~~~~~~~~~~~~~~~~~~~~~~~~~~~~~~~~~~~~~~~~~s                &0& 1 & 2 & 3 &  4 & 5 & 6 & 7  \\
\hline
v_1                                                               &1&0 & 1 & 0  & 0 & 0 & 0 & 0  \\
v_2                                                               &1&1 & 0 & 0  & 0 & 0 & 0 & 0  \\
v_3                                                               &1&0 & 0 & 0  & 1 & 0 & 0 & 0  \\
v_4                                                               &1&0 & 0 & 1  & 0 & 0 & 0 & 0  \\
v_5                                                               &1&0 & 1 & 0  & 0 & 0 & 0 & 0  \\
v_6                                                               &1&0 & 0 & 1  & 0 & 0 & 0 & 0  \\
\hline
v_1\times  v_2                                                    &1& 1 & 0 & 1 & 0 & 0 & 0 & 0\\
v_3 \times v_4                                                    &1& 0 & 0 & 1 & 0 & 0 & 0 & 1\\
v_5 \times v_6                                                    &1& 0 & 0 & 1 & 0 & 1 & 0 & 0\\
(v_1\times  v_2) \cup (v_3 \times v_4)                            &1& 1 & 0 & 1 & 1 &0  & 1 & 1 \\
((v_1\times  v_2) \cup (v_3 \times v_4)) \times(v_5 \times v_6)   &1& 0 & 0 & 1 & 0 &1  &1 & 0\\
\hline
\end{array}
$$
\end{table}

\begin{theorem}\label{th-co-ssg-sp}
SSG can be solved in  series-parallel digraphs  with $n$
vertices and $m$ arcs
in $\bigo(n\cdot c^2+n^{2.3729})$ time  and $\bigo(n\cdot c)$ space.
\end{theorem}

\begin{proof}
Let $G$ be some  series-parallel digraph. By Lemma \ref{lem2}
we can use the transitive reduction of $G$, which can be computed in
$\bigo(n^{2.3729})$ time by \cite{Leg14}. 
\end{proof}

\subsection{Subset sum  with weak digraph constraint (SSGW)}

Next, we consider SSGW on minimal series-parallel digraph. 
In order to get useful informations about the sinks within a
solution, we use an extended data structure.
We consider an instance of SSGW such that $G=(A,E)$ is a minimal series-parallel digraph which
is given by some  msp-expression  $X$.
For some subexpression $X'$ of $X$ let $H(X',s,s')=1$ if there is a solution
$A'$ in the graph defined by $X'$  satisfying (\ref{cap}) and (\ref{wdc})  
such that $s(A')=s$ and the sum of sizes of
the sinks in $A'$ is $s'$,  otherwise let $H(X',s,s')=0$. We denote by $i(X)$
the sum of the sizes of all sinks in $\g(X)$.

\begin{lemma}\label{le-ssgw-sp} Let $0\leq s,s' \leq c$.
\begin{enumerate}
\item $H(a_j,s,s')=1$ if and only if $s=s'=0$ or $s_j=s=s'$.

In all other cases  $H(a_j,s,s')=0$.

\item $H(X_1\cup X_2,s,s')=1$, if  and only if 
there are  $0\leq s_1\leq s$, $0\leq s_2\leq s$, $0\leq s'_1\leq s'$, $0\leq s'_2\leq s'$,
such that $s_1+s_2=s$, $s'_1+s'_2=s'$, $H(X_1,s_1,s'_1)=1$,
and $H(X_2,s_2,s'_2)=1$.

In all other cases  $H(X_1\cup X_2,s,s')=0$.

\item $H(X_1\times X_2,s,s')=1$, if  and only if 
\begin{itemize}
\item
$0\leq s \leq  s(X_2)$\footnote{The value $s=s'=0$ is for choosing an empty 
solution in $\g(X_1\times X_2)$. The values $s>s'=0$ are for choosing a 
solution without sinks in $\g(X_1\times X_2)$} and $0\leq s' \leq s(X_2)$,  
such that $H(X_2,s,s')=1$   or

\item
there are  $1\leq s_1 \leq s(X_1)$ and $1\leq s'_1 < i(X_1)$,
such that $s_1=s$, $0=s'$, and $H(X_1,s_1,s'_1)=1$, or

\item
there are   $1\leq s_1 \leq s(X_1)$, 
such that $s_1+s(X_2)=s$, $i(X_2)=s'$, and $H(X_1,s_1,i(X_1))=1$, or

\item
there are    $1\leq s_1 \leq s(X_1)$, $1\leq s'_1 < i(X_1)$, $1\leq s_2 \leq s(X_2)$, and 
$1\leq s'_2 \leq s(X_2)$,
such that $s_1+s_2=s$, $s'_2=s'$, $H(X_1,s_1,s'_1)=1$, and  $H(X_2,s_2,s'_2)=1$.
\end{itemize}
In all other cases  $H(X_1\times X_2,s,s')=0$.
\end{enumerate}
\end{lemma}

\begin{proof}
We show the correctness of the stated equivalences. Let $0\leq s, s' \leq c$.
\begin{enumerate}
\item 
The only possible solutions in $\g(a_j)$  are $\emptyset$ and $\{a_j\}$ which have size 
$0$ and $s_j$, respectively. Further, a single vertex corresponds to
a sink.

\item 
If $H(X_1  \cup X_2,s,s')=1$, then  by Lemma \ref{induced-sd3} there are $s_1$, $s_2$ and $s'_1$, $s'_2$ 
such that $s_1+s_2=s$, $s'_1+s'_2=s'$ and solutions in $\g(X_1)$  
and in $\g(X_2)$ which guarantee $H(X_1,s_1,s'_1)=1$ and $H(X_2,s_,s'_2)=1$.

\medskip
Further, for every $0\leq s_1\leq s$, $0\leq s_2\leq s$, $0\leq s'_1\leq s'$, $0\leq s'_2\leq s'$,
such that $s_1+s_2=s$, $s'_1+s'_2=s'$, $H(X_1,s_1,s'_1)=1$,
and $H(X_2,s_2,s'_2)=1$, it holds  that $H(X_1  \cup X_2,s,s')=1$
since we do not create any new edges by the parallel composition.

\item
If $H(X_1  \times X_2,s,s')=1$,  then we distinguish four cases.
If the solution of size $s$ and sink size $s'$ in $\g(X_1\times X_2)$ contains no
vertices of  $\g(X_1)$, then by Lemma \ref{induced-sd4} there is a solution in $\g(X_2)$ which
guarantees $H(X_2,s,s')=1$.  

If the solution of size $s$ and sink size $s'$ in $\g(X_1\times X_2)$ contains only 
vertices of  $\g(X_1)$ but not all sinks of $\g(X_1)$, then by Lemma \ref{induced-sd4} there is a solution in $\g(X_1)$ which
guarantees $H(X_1,s,s')=1$.  

If the solution $A'$ of size $s$ and sink size $s'$ in $\g(X_1\times X_2)$ contains 
all sinks of $\g(X_1)$,  the series
composition and the weak digraph constraint (\ref{wdc}) imply 
that the set $A'$ has to be extended by all sources of $\g(X_2)$.
After ignoring the sources of $\g(X_2)$ (because the graph is acyclic), 
there must exist new sources, which have to be contained in $A'$, since 
all their predecessors were sources in the original graph and so on. 
Thus, set $A'$ contains all vertices of $X_2$ and by Lemma \ref{induced-sd4}
there is a solution in $\g(X_1)$ which
guarantees $H(X_1,s-s(X_2),i(X_1))=1$.

If the solution $A'$ of size $s$ and sink size $s'$ in $\g(X_1\times X_2)$ contains 
vertices of  $\g(X_1)$ but not all sinks of $\g(X_1)$
and vertices of  $\g(X_2)$, then 
by  Lemma \ref{induced-sd4}  there are  
$s_1,s'_1$ and $s_2,s'_2$ such that $s_1+s_2=s$, $s'_2=s'$ and solutions in $\g(X_1)$  
and in $\g(X_2)$ which guarantee $H(X_1,s_1,s'_1)=1$ and $H(X_2,s_2,s'_2)=1$. 

\medskip
Further, the solutions of size $0\leq s \leq  s(X_2)$  from $\g(X_2)$ 
remain feasible in $\g(X_1\times X_2)$ since the solutions from $\g(X_2)$ do not contain
any predecessors of vertices from $\g(X_1)$ in $\g(X_1\times X_2)$.

The solutions from $\g(X_1)$  which
do not contain all sinks of $X_1$, i.e.\ $1\leq s'_1 < i(X_1)$ remain feasible in $\g(X_1\times X_2)$, but
the sizes of sinks have to be changed to $0$ since these sinks are no longer sinks  
in the  $\g(X_1\times X_2)$.

Next we consider solutions $A'$ from $\g(X_1)$ which contain all sinks of $\g(X_1)$, 
i.e.\ $s'= i(X_1)$. As mentioned above, the series
composition and the weak digraph constraint (\ref{wdc}) imply 
that the set $A'$ 
has  to be extended by all vertices of $X_2$.
The sizes of sinks have to be changed to $i(X_2)$, since all sinks of $X_2$ are 
also sinks  in the  $\g(X_1\times X_2)$.

Further, we can combine solutions of  size $1\leq s_1 \leq  s(X_1)$ from $\g(X_1)$, which 
do not contain all sinks of $X_1$, i.e.\ $1\leq s'_1 < i(X_1)$, and 
solutions of  size $1\leq s_2 \leq  s(X_2)$ from $\g(X_2)$,
to a solution of size $s_1+s_2$ and sizes of sinks $s'_2$ in $\g(X_1\times X_2)$.
\end{enumerate}
This shows the statements of the lemma.
\end{proof}

\begin{corollary}\label{cor2-sp}
There is a solution with sum $s$  for some
instance of SSGW such that $G$ is a minimal series-parallel digraph which
is given by some msp-expression $X$
if and only if $H(X,s,s')=1$. Therefore, $OPT(I)=\max\{s \mid H(X,s,s')=1\}$.
\end{corollary}

\begin{theorem}\label{th-co2-ssg-sp}
SSGW  can be solved in minimal series-parallel digraphs  with $n$
vertices and $m$ arcs
in  $\bigo(n\cdot c^4+m)$ time and $\bigo(n\cdot c^2)$ space.
\end{theorem}

\begin{proof}
Let $G=(V,E)$  be a minimal series-parallel digraph and $T$ be an
msp-tree for $G$ with root $r$.
For some vertex $u$ of $T$ we denote by $T_u$
the subtree rooted at $u$ and $X_u$ the msp-expression defined by $T_u$.
In order to solve the SSGW problem for an instance $I$
graph $G$, we traverse  msp-tree $T$ in a bottom-up order.
For every vertex $u$ of $T$ and $0\leq s,s' \leq c$ we compute $H(X_u,s,s')$
following the rules given in Lemma \ref{le-ssgw-sp}. 
By Corollary \ref{cor2-sp} we can solve our
problem by $H(X_r,s,s')=H(X,s,s')$.

An msp-tree $T$ can be computed in $\bigo(n+m)$ time from
a  minimal series-parallel digraph with $n$ vertices and $m$ arcs,  see \cite{VTL82}.
All $s(X_i)$ and all $i(X_i)$ can be precomputed in $\bigo(n)$ time.
Our rules given in Lemma \ref{le-ssgw-sp} show the following running times.
\begin{itemize}
\item
For every $a_j\in A$ and every $0\leq s,s' \leq c$ value
$H(a_j,s,s')$ is computable in $\bigo(1)$ time.

\item
For every $0\leq s,s' \leq c$, every
$H(X_1  \cup X_2,s,s')$   can be computed in $\bigo(c^2)$ time from
$H(X_1,s_1,s'_1)$ and $H(X_2,s_2,s'_2)$.

\item For every $0\leq s,s' \leq c$,
every $H(X_1  \times X_2,s,s')$ can be computed
in  $\bigo(c^2)$ time from $H(X_1,s_1,s'_1)$, $H(X_2,s_2,s'_2)$,
and $i(X_1)$.
\end{itemize}
Since we have $n$ leaves and $n-1$ inner vertices in $T$,
the running time is in $\bigo(nc^4+m)$.
\end{proof}

\begin{example}\label{ex-ssg-spg}
We consider the SSGW instance $I$ with $n=6$ items using $\g(X)$
defined 
by the expression in  (\ref{eq-ori-c4msp2}), $c=7$, and the following sizes.
$$
\begin{array}{l|llllll}
j  &  1 & 2 & 3 & 4 & 5 & 6\\
\hline
s_j&  2 & 1 & 4  & 3 & 2 & 3
\end{array}
$$

The rules  given in  Lemma \ref{le-ssgw-sp} lead to
the values in Table \ref{tab-ssg-spg}.
Thus the optimal solution is $\{v_3,v_4\}$ with $OPT(I)=7$.
\end{example}

\begin{table}[h!]
\centering
\caption{Table for Example \ref{ex-ssg-spg}\label{tab-ssg-spg}}


\rotatebox{90}{\footnotesize
$
\begin{array}{l|llllllll|llllllll|llllllll|llllllll|}
  & \multicolumn{32}{c|}{H(X',s,s')}   \\
\hline
      & \multicolumn{8}{c|}{s'=0}   & \multicolumn{8}{c|}{s'=1}     & \multicolumn{8}{c|}{s'=2}   & \multicolumn{8}{c|}{s'=3}\\
X'~~~~~~~~~~~~~~~~~~~~~~~~~~~~~~~~~~~~~~~~~~~~~~~s              &0&1&2&3&4&5&6&7  & 0& 1 & 2 & 3 &  4 & 5 & 6 & 7   & 0&  1 & 2 & 3 &  4 & 5 & 6 & 7 & 0& 1 & 2 & 3 &  4 & 5 & 6 & 7   \\
\hline
v_1                                                             &1 &0&0&0&0&0&0&0& 0&0&0&0&0&0&0&0& 0&0&1&0&0&0&0&0& 0&0&0&0&0&0&0&0\\
v_2                                                             &1 &0&0&0&0&0&0&0& 0&1&0&0&0&0&0&0& 0&0&0&0&0&0&0&0& 0&0&0&0&0&0&0&0\\
v_3                                                             &1 &0&0&0&0&0&0&0& 0&0&0&0&0&0&0&0& 0&0&0&0&0&0&0&0& 0&0&0&0&0&0&0&0\\
v_4                                                             &1 &0&0&0&0&0&0&0& 0&0&0&0&0&0&0&0& 0&0&0&0&0&0&0&0& 0&0&0&1&0&0&0&0\\
v_5                                                             &1 &0&0&0&0&0&0&0& 0&0&0&0&0&0&0&0& 0&0&1&0&0&0&0&0& 0&0&0&0&0&0&0&0\\
v_6                                                             &1 &0&0&0&0&0&0&0& 0&0&0&0&0&0&0&0& 0&0&0&0&0&0&0&0& 0&0&0&1&0&0&0&0\\

\hline
v_1\times  v_2                                                  & 1&0&0&0&0&0&0&0& 0&1&0&1&0&0&0&0& 0&0&0&0&0&0&0&0& 0&0&0&0&0&0&0&0\\
v_3 \times v_4                                                  & 1&0&0&0&0&0&0&0& 0&0&0&0&0&0&0&0& 0&0&0&0&0&0&0&0& 0&0&0&1&0&0&0&1\\
v_5 \times v_6                                                  & 1&0&0&0&0&0&0&0& 0&0&0&0&0&0&0&0& 0&0&0&0&0&0&0&0& 0&0&0&1&0&1&0&0\\
(v_1\times  v_2) \cup (v_3 \times v_4)                          & 1&0&0&0&0&0&0&0& 0&1&0&1&0&0&0&0& 0&0&0&0&0&0&0&0& 0&0&0&1&0&0&0&1\\
((v_1\times  v_2) \cup (v_3 \times v_4)) \times(v_5 \times v_6) & 1&1&0&1&0&0&0&1& 0&0&0&0&0&0&0&0& 0&0&0&0&0&0&0&0& 0&0&0&1&1&1&1&0\\
\hline
\\
\\
\\
  & \multicolumn{32}{c|}{H(X',s,s')}  \\
\hline
      &  \multicolumn{8}{c|}{s'=4}     & \multicolumn{8}{c|}{s'=5}   & \multicolumn{8}{c|}{s'=6}  & \multicolumn{8}{c|}{s'=7} \\
X'~~~~~~~~~~~~~~~~~~~~~~~~~~~~~~~~~~~~~~~~~~~~~~~s              & 0& 1 & 2 & 3 &  4 & 5 & 6 & 7 & 0& 1 & 2 & 3 &  4 & 5 & 6 & 7 & 0& 1 & 2 & 3 &  4 & 5 & 6 & 7  & 0& 1 & 2 & 3 &  4 & 5 & 6 & 7 \\
\hline
v_1                                                             & 0&0&0&0&0&0&0&0& 0&0&0&0&0&0&0&0& 0&0&0&0&0&0&0&0& 0&0&0&0&0&0&0&0\\
v_2                                                             & 0&0&0&0&0&0&0&0& 0&0&0&0&0&0&0&0& 0&0&0&0&0&0&0&0& 0&0&0&0&0&0&0&0\\
v_3                                                             & 0&0&0&0&1&0&0&0& 0&0&0&0&0&0&0&0& 0&0&0&0&0&0&0&0& 0&0&0&0&0&0&0&0\\
v_4                                                             & 0&0&0&0&0&0&0&0& 0&0&0&0&0&0&0&0& 0&0&0&0&0&0&0&0& 0&0&0&0&0&0&0&0\\
v_5                                                             & 0&0&0&0&0&0&0&0& 0&0&0&0&0&0&0&0& 0&0&0&0&0&0&0&0& 0&0&0&0&0&0&0&0\\
v_6                                                             & 0&0&0&0&0&0&0&0& 0&0&0&0&0&0&0&0& 0&0&0&0&0&0&0&0& 0&0&0&0&0&0&0&0\\

\hline
v_1\times  v_2                                                  & 0&0&0&0&0&0&0&0& 0&0&0&0&0&0&0&0& 0&0&0&0&0&0&0&0& 0&0&0&0&0&0&0&0\\
v_3 \times v_4                                                  & 0&0&0&0&0&0&0&0& 0&0&0&0&0&0&0&0& 0&0&0&0&0&0&0&0& 0&0&0&0&0&0&0&0\\
v_5 \times v_6                                                  & 0&0&0&0&0&0&0&0& 0&0&0&0&0&0&0&0& 0&0&0&0&0&0&0&0& 0&0&0&0&0&0&0&0\\
(v_1\times  v_2) \cup (v_3 \times v_4)                          & 0&0&0&0&1&0&1&0& 0&0&0&0&0&0&0&0& 0&0&0&0&0&0&0&0& 0&0&0&0&0&0&0&0\\
((v_1\times  v_2) \cup (v_3 \times v_4)) \times(v_5 \times v_6) &  0&0&0&0&0&0&0&0& 0&0&0&0&0&0&0&0& 0&0&0&0&0&0&0&0& 0&0&0&0&0&0&0&0\\
\hline

\end{array}
$
}
\end{table}

\section{Conclusions and outlook}\label{sec-con}

The presented methods allow us to solve SSG and SSGW with
digraph constraints given by directed co-graphs and
(minimal) series-parallel digraphs in pseudo-polynomial time.
 
In contrast to \cite{GMT18} we did not consider null sizes.
This allows us  to verify whether a solution
consists of all vertices or contains all sinks of a subgraph
by using the sum of the sizes of the corresponding items.
SSG and SSGW using null sizes can also be solved in pseudo-polynomial time 
on directed co-graphs and (minimal) series-parallel digraphs by additional
counting the number of vertices or sinks within a SSGW solution.

For future work it could be interesting to find a solution for SSGW for
series-parallel digraphs in general. Example \ref{ex-ssgw-tr} shows that Lemma
\ref{lem2}
and the recursive structure of  minimal series-parallel digraphs cannot
be used in this case.
 
It remains  to analyze whether the shown
results also hold for other graph classes. Therefore one could consider
edge series-parallel digraphs from \cite{VTL82}. Further,
it remains to look at more  general graph classes, such as graphs
of bounded directed clique-width.
Directed clique-width  measures the difficulty of
decomposing a graph into a special tree-structure and was
defined by  Courcelle and Olariu in \cite{CO00}.
An alternative parameter is directed tree-width defined in \cite{JRST01}.
Since in the directed case bounded directed tree-width does not
imply bounded directed clique-width, solutions for
subset sum problems with digraph constraints of bounded directed
tree-width are interesting as well.
 
Furthermore, it could be useful to consider related problems.
These include the two minimization problems which are
introduced in \cite{GMT18}  by adding a maximality constraint
to SSG and SSGW. Moreover, a generalization of the results for SSG to the
partially ordered knapsack problem \cite{JN83,KP04} is still open.

\section*{Acknowledgements} \label{sec-a}

The work of the second and third author was supported
by the Deutsche
Forschungsgemeinschaft (DFG, German Research Foundation) -- 388221852


\bibliographystyle{alpha}

\newcommand{\etalchar}[1]{$^{#1}$}

\end{document}